\documentclass{llncs}

\usepackage{verbatim}
\usepackage{fancyvrb}
\usepackage{amsmath}
\usepackage{amssymb}
\usepackage{graphicx,caption}
\usepackage{verbatim}
\usepackage[colorlinks=true, allcolors=blue]{hyperref}

\newtheorem{goal}{Goal}
\newtheorem{assumptions}{Assumption}
\newtheorem{context}{Session Context}

\newcommand{\myno}{$\times$}
\newcommand{\myheader}[1]{{ \textbf{\textit{{#1}} }}}

\newcommand{\mysaveA}{\vspace*{-6pt}} 
\newcommand{\mysaveAA}{\vspace*{-3pt}} 
\newcommand{\mysaveB}{\vspace*{-4pt}} 
\newcommand{\mysaveC}{\vspace*{-2pt}} 

\title{Cryptographic Binding Should Not Be Optional:\\
A Formal-Methods Analysis of\\
FIDO UAF Channel Binding}

\author{Enis Golaszewski\inst{1} \and 
Alan T. Sherman\inst{1} \and 
Edward Zieglar\inst{2} \and \\
Jonathan D. Fuchs\inst{1} \and
Sophia Hamer\inst{1} 
}

\institute{Cyber Defense Lab, University of Maryland, Baltimore County (UMBC),\\ 
1000 Hilltop Circle, Baltimore MD 21228, USA \\
\email{\{golaszewski, sherman, jfuchs2, chamer1\}@umbc.edu }\and 
National Security Agency, 9800 Savage Road, 20755 Fort George G. Meade, USA 
\email{evziegl@uwe.nsa.gov}}

\date{July 19, 2023}

\begin{document}

\maketitle

\begin{abstract} 
As a case study in cryptographic binding, we present a formal-methods analysis of 
the  cryptographic channel binding mechanisms in the 
{\it Fast IDentity Online (FIDO) Universal Authentication Framework (UAF)} 
authentication protocol, which
seeks to reduce the use of traditional passwords in favor of authentication devices.
First, we show that UAF's channel bindings fail to mitigate protocol interaction 
by a \textit{Dolev-Yao} adversary, enabling the adversary to transfer the server's authentication challenge to alternate sessions of the protocol.
As a result, in some contexts, the adversary can masquerade as 
a client and establish an authenticated session with a server (e.g., possibly a bank server).
Second, we implement a proof-of-concept man-in-the-middle attack against eBay's open source FIDO UAF implementation.
Third, we propose and formally verify improvements to UAF.
The weakness we analyze is similar to the vulnerability discovered in the Needham-Schroeder protocol over 25 years ago. 
That this vulnerability appears in the FIDO UAF standard
highlights the strong need for protocol designers to bind messages properly and to analyze their designs with formal-methods tools.
To our knowledge, we are first to carry out a formal-methods analysis of channel binding in UAF and first to exhibit details of an attack on UAF that exploits the weaknesses of UAF's channel binding.
Our case study illustrates the importance of cryptographically binding context to protocol messages to prevent an adversary from misusing messages out of context.
\end{abstract}



{\mysaveA}
\section{Introduction}

42 years apart, the 1978 {\it Needham-Schroeder (NS)} 
public-key protocol~\cite{Needham1978}
and the 2020 {\it Fast IDentity Online (FIDO) Universal Authentication Framework (UAF)} 1.2 specification~\cite{baghdasaryan2020fido}
share a common flaw: they both fail to cryptographically bind a sensitive value to its source.
To protect against a malicious adversary, including
a {\it Dolev-Yao (DY)} network intruder~\cite{dolevyao}, 
cryptographic binding is necessary.
Without correct cryptographic binding, an adversary may violate a protocol's security goals by copying sensitive cryptographic values between protocol sessions to execute {\it man-in-the-middle (MitM)} attacks.
Mitigating such attacks requires emerging protocol standards to apply cryptographic bindings correctly and to perform formal-methods analyses to verify the correctness of these bindings.
Using formal methods, we analyze a channel-binding flaw of FIDO UAF, 
demonstrate a resulting MitM attack against eBay's open-source UAF server, and
suggest and formally verify mitigations.

Our analysis reveals MitM attacks on the UAF attestations that 
enable an adversary to masquerade as a client, for example, 
when establishing a session with a bank server.
A version of this attack follows from any one of the following weaknesses in UAF:
(1)~Even when performing channel binding, the server might not bind the challenge in a manner that enables the client to authenticate the challenge.
(2)~The server can selectively accept incorrect channel bindings, potentially accepting attestations from malicious protocol sessions.
(3)~The standard makes channel bindings optional, creating circumstances in which there is no cryptographic binding between the client's attestation and the server's protocol session.
Exploiting these weaknesses, an adversary can trick a legitimate client into acting as a confused deputy~\cite{hardy1988confused}, producing attestations for a legitimate server's challenge, which the legitimate server accepts. 

Cryptographic \textit{channel binding} is an important type of cryptographic binding which binds sensitive values to an underlying authenticated communication channel.
A common type of channel binding is \textit{Transport Layer Security (TLS)} channel binding~\cite{rescorla2001ssl}, 
which, through various methods, binds sensitive values to properties of an end-to-end encrypted TLS channel.
The UAF specification recommends TLS channel binding to mitigate MitM attacks, in which an adversary might fraudulently authenticate as a legitimate user to steal the user's private information, engage in transactions on the user's behalf, or usurp the user's identity.
For emerging protocol standards, channel binding is a critical tool for resisting MitM attacks and other undesirable structural flaws.


Formal-methods protocol analysis is a powerful technique for proving the correctness of a protocol's security properties or, by discovering counterexamples, identifying security flaws.
This technique lends itself well to automatic tools (e.g., theorem provers or model checkers---see Section~\ref{ssec:formal-methods}), which verify the security properties of a cryptographic protocol specified in a formal language.
Generally, these tools assume \textit{symbolic} models with perfect cryptography, and evaluate protocols against a DY adversary to reveal security flaws.
Security flaws are subtle, particularly when an adversary manipulates messages between an unbounded number of protocol instances, and can be difficult to identify manually.


We present a formal-methods analysis of channel bindings in the 2020 FIDO UAF 1.2 protocol,
which seeks to reduce the use of traditional passwords in favor of authentication devices.
Using the \textit{Cryptographic Protocol Shapes Analyzer (CPSA)}~\cite{liskov2016cryptographic}---which is one of several tools specialized for formal-methods analysis of protocols---we
model several variations of FIDO UAF in the 
mathematical abstraction of \textit{strand space} (see Section~\ref{ssecstrands}) and
apply CPSA to explore all essentially different possible executions of each model.
We do so separately from the perspectives of the client and the server, for
specified assumptions about how certain important cryptographic values (e.g., keys, nonces)
were generated and who knows them.
This exploration includes possible interactions with a DY adversary. 
CPSA either exhaustively proves that no adverse interaction is possible, or CPSA
finds an attack against the model.  
We summarize our results by stating and proving, or disproving, UAF security
properties, focusing on UAF's declared security goals.

While powerful, formal-methods analysis does not address all aspects of protocol security.
For example, this type of analysis will not uncover errors in implementation, flaws resulting from using weak cryptographic primitives, or applications of protocols to inappropriate settings.
For these reasons, in our review of eBay's implementation of FIDO UAF, 
we also carefully studied their source code.
Despite the limitations of formal-methods analysis, protocols---including UAF---would
usually benefit from such analyses.

\myheader{FIDO UAF.}
UAF is an attractive emerging standard and the subject of a growing number of studies, including formal-methods studies (see Section~\ref{sec:previous}).
The FIDO Alliance counts among its members many recognizable technology giants, financial institutions, retailers, government institutes, and standards bodies from nations throughout the world~\cite{fidomembers}.
These members include Google, Apple, Microsoft, and Amazon.
UAF is a complicated standard, specified across over a dozen documents, that attempts to support a large number of technologies and use cases.
To our knowledge, FIDO developed the standard without any formal verification, resulting in several studies (including this one) identifying flaws.
With UAF replacing many traditional password-based systems, it is vital that we analyze and improve the standard's cryptographic bindings.

A core component of UAF is the \textit{authenticator}---a networked device implementing a local authentication mechanism, such as facial recognition, voice detection, or PIN entry.
UAF specifies four protocols: registration, authentication, transaction confirmation, and de-registration.
In registration, a \textit{client} pairs an authenticator with a server 
(the ``\textit{relying party}'') for subsequent use in the other protocols.
In authentication, transaction confirmation, and de-registration, the client authenticates to a relying party by presenting 
\textit{attestations} (assertions) of a user's identity from a subset of paired authenticators (see Figure~\ref{fig:uaf_arch}).
To associate an authenticator's assertion with a client's specific request, the relying party issues a cryptographic challenge, which the authenticator signs.

In FIDO UAF, the relying party's challenge, and thus the resulting assertion, may optionally bind to the session's underlying TLS channel using one of several TLS channel bindings.
By shedding light on how to channel bind correctly, including the choice of binding method and to which data to bind, our study of UAF's channel binding benefits UAF and future protocol standards.
We introduce UAF's available channel bindings in Section~\ref{ssec:binding} and discuss FIDO's decision to make these bindings optional in Section~\ref{sec:optional_binding}.

\myheader{Our Work.}
To our knowledge, we are first, including using formal methods, to analyze whether the available TLS channel bindings adequately address the UAF security goals.  In particular, prior work by Feng et al.\ \cite{feng2021formal} does not perform such analysis.
 

To analyze channel binding in UAF, we model and analyze variations of the registration and authentication protocols using CPSA.
For each protocol, we model five available channel bindings (including no channel binding) using three different versions of TLS: TLS 1.2 with RSA key exchange (TLS1.2-RSA), TLS 1.2 with DH key exchange (TLS1.2-DH), and TLS 1.3.
Additionally, we model each protocol using the TLS 1.3 
``\textit{TLS-exporter}'' channel binding'~\cite{rfc9266}, 
which UAF does not officially support.
Table~\ref{tbl:models} lists our CPSA models of the UAF protocols.
To reflect the interactions between the lower-level TLS connections and the higher-level UAF messages, 
each CPSA model incorporates existing strand-space models of TLS 1.2 and TLS 1.3.
We model and analyze these different binding scenarios to illustrate the binding properties necessary to achieve UAF security goals under a variety of cryptographic assumptions.

As a proof of concept, we implement our MitM attack against eBay's FIDO-certified, open-source UAF server \cite{eBayFIDO}, which does not implement channel binding (see Section~\ref{sec:attack}).
On August 13, 2023, we notified eBay's GitHub maintainers of this vulnerability, but 
they never responded, and eBay appears to have abandoned their UAF implementation.
To carry out this attack, an adversary exploits inadequate binding of the server's challenge to specific sessions of UAF authentication and registration protocols, producing a protocol interaction in which the honest client-authenticator pair generates attestations for the adversary's malicious sessions.
This attack demonstrates an example of a harmful protocol interaction resulting from inadequate cryptographic binding and verification.

Our contributions include: 
(1)~A formal-methods analysis of UAF~1.2's optional channel bindings, including bindings for TLS-1.2 (with RSA or DH key exchange) and TLS-1.3. 
Table~\ref{tbl:goals} summarizes the results of our analysis.
(2)~A structural weakness and resulting ``challenge-reissuing'' attack against several FIDO UAF channel bindings, exploiting inadequate binding of the server's challenge to a particular session (Section~\ref{sec:VAR}).
(3)~Details of a MitM attack in which an adversary reissues challenges to trick a client and authenticator pair to act as confused deputies to authenticate the adversary 
(Section~\ref{ssecattacks}).
(4)~A demonstration of the MitM attack against eBay's open-source UAF server 
(Section~\ref{sec:attack}).
(5)~Recommendations to improve channel binding in UAF 
(Section~\ref{sec:rec}).
We include artifacts of our work, including all of our CPSA and attack source code, and 
make these artifacts available on GitHub~\cite{PALgithub}.

In the rest of this paper, we 
(1)~introduce the UAF protocols and their channel binding mechanisms,
(2)~discuss previous work,
(3)~present our adversarial model, CPSA models, security goals, and security analysis,
(4)~identify resulting vulnerabilities, attacks, and risks,
(5)~describe and implement an attack against eBay's UAF implementation,
and (6)~recommend improvements to the UAF standard.
Several appendices include additional details, including
a high-level overview of UAF's primary structural weakness and 
three examples of resulting possible attacks;
background on cryptographic binding, strand spaces, and CPSA;
ladder diagrams of the UAF protocols; 
details about our models, security goals, formal-methods analyses; and
selected snippets of our CPSA source code.


{\mysaveA}
\section{FIDO}
\label{sec:fido}

We introduce FIDO UAF and discuss the framework's cryptographic channel binding, authentication protocol, and registration protocol.

{\mysaveB}
\subsection{FIDO UAF}
{\mysaveC}
{\mysaveC}

In 2013, the FIDO Alliance proposed UAF~\cite{baghdasaryan2020fido}, 
an open standard supplanting passwords in favor of 
\textit{authenticator devices}, such as mobile phones, which implement local authentication mechanisms (e.g., biometric, PIN).
Major design goals of UAF include enabling \textit{relying parties} (services seeking to authenticate users), 
to specify to users which types of authenticators they will accept.
As of September 2025, the FIDO website lists 508 certified implementations of UAF~\cite{fidocert}, including the eBay implementation that we attack in Section~\ref{sec:attack}.

UAF's design seeks to reduce the use of traditional passwords, which FIDO indicates as problematic due to weak password choices, password reuse, management of many passwords, and phishing attacks.
FIDO identifies passwords as responsible for over 80\% of all data breaches, and promotes UAF as a potential mitigation.
Additionally, relying parties that authenticate users using UAF do not store passwords or password hashes, mitigating common attacks in which the adversary steals password databases. 


\begin{figure}[t]
\includegraphics[width=\textwidth]{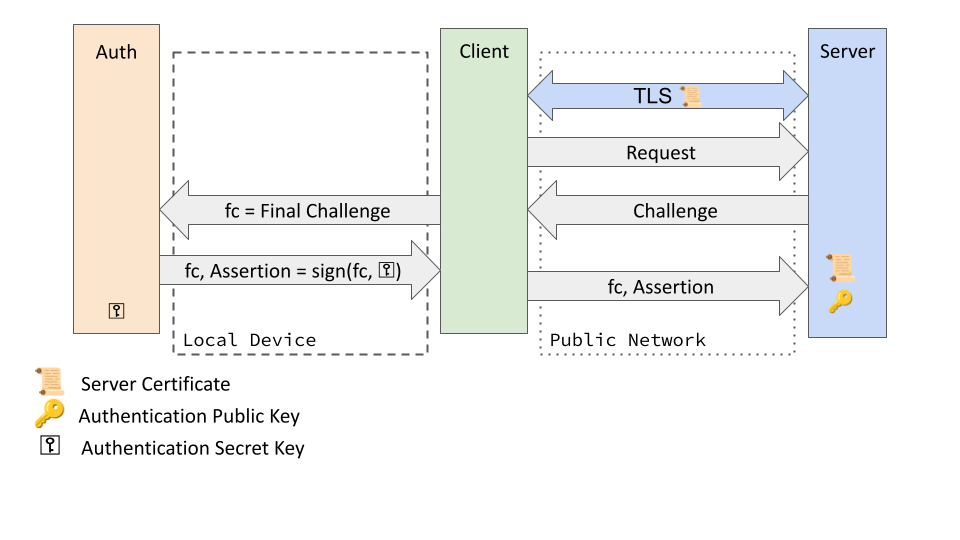}
\caption{
Architecture of UAF authentication using a pre-registered, embedded authenticator.
The server issues a challenge to the client, who incorporates that challenge into a final challenge and transmits it to the authenticator (Auth).
Using the authentication secret key, the authenticator signs an assertion, which the client transmits to the server to complete authentication.
The server verifies the assertion using the authentication public key that it obtains when the client registers the authenticator.
The attestation process in UAF registration reflects a similar structure.
In our analysis, we assume that the adversary cannot manipulate local device communication between the client and their authenticator.
}
\label{fig:uaf_arch}
\end{figure}


The FIDO Security Reference~\cite{baghdasaryan2022fido} specifies informal security goals for UAF network protocols, including strong user authentication (SG-1), forgery resistance (SG-11), parallel session resistance (SG-12), and forwarding resistance (SG-13), for which the reference lists channel binding as a security measure.
Despite listing channel binding as a security measure, 
to encourage adoption of the framework and to maximize compatibility with existing systems,
\textit{FIDO made channel binding optional} in the UAF standard.
In Section~\ref{sec:goals}, we formalize these security
goals as an ``injective agreement''~\cite{lowe1997hierarchy} on session context, 
and, in Section~\ref{sec:optional_binding}, 
we discuss FIDO's decision to make channel binding optional.

UAF protocols build on the principle of \textit{attestation}: clients prove their identities to relying parties by presenting one or more assertions of their identity signed by pre-registered authenticators.
Authenticators sign these assertions with 
asymmetric authentication keys that they generate when registering with a relying party.
Additionally, the authenticators sign \textit{endorsement keys} using \textit{attestation signing keys}, which the authenticators back by presenting \textit{attestation certificates} originating from device manufacturers.
By relying on authenticator attestations, UAF shifts authentication from traditional passwords (``something you know") to a client's possession of their registered authenticators (``something you have").
Figure~\ref{fig:uaf_arch} illustrates a high-level overview of the assertion process during UAF authentication.

UAF presents challenges to developers and researchers because it is 
complex---the standard comprises several protocols and variations spanning at least 12 lengthy documents---and specifies security goals only informally.
This complexity arises from the framework's effort to be universal: the framework attempts to support many types and models of authenticators, incorporate backwards compatibility with legacy authentication methods, and enable relying parties to specify custom levels of security and authenticator support.
To manage UAF's complexity, implementations may forgo optional features such as cryptographic channel binding, and formal-methods analyses may abstract away certain operations.
These measures may result in unexpected vulnerabilities.

{\mysaveB}
\subsection{Channel Binding}
{\mysaveC}
\label{ssec:binding}

To mitigate potential MitM attacks, UAF specifies an optional channel binding that binds a UAF session's challenge parameters to the underlying TLS channel between the client and the server.
UAF implementations may select from several TLS channel binding standards: 
\textit{token binding}~\cite{rfc8471}, 
\textit{channel ID binding}~\cite{balfanz2013channel}, 
\textit{server endpoint binding}~\cite{rfc5929}, and 
\textit{server certificate binding} 
(could be done in a manner similar to client certificate binding in OAuth~\cite{lodderstedt2025rfc}).
In token binding and channel ID binding, the client generates an asymmetric key pair, shares the public key with the server using extensions in the TLS handshake, and binds subsequent UAF challenge parameters to this public key.
When binding to the server endpoint, the client binds the challenge to a hash of the server's certificate.
When binding to the server certificate, the client binds the challenge
to the server's certificate.

While the phrase ``channel binding'' suggests binding to a specific TLS \textit{session}, 
each of the UAF's channel bindings binds only to one of two \textit{endpoints}: 
the client or the server.
The channel bindings are valid for multiple TLS sessions between endpoints, such as renegotiated sessions, but potentially enable bound UAF session challenge parameters to migrate between different UAF sessions sharing the same endpoints (see Section~\ref{sec:analysis}).
We analyze UAF with the framework's supported bindings and a newer channel binding 
TLS-exporter~\cite{rfc9266}, which binds to a specific TLS session but is unsupported by the current UAF standard.
To our knowledge, TLS 1.3 officially supports only TLS-exporter as a valid binding.

UAF allows a server to continue the protocol even when it is unable to verify bindings.
For example, the server might accept a binding from a client that binds to a TLS connection between the client and a TLS proxy, rather than between the client and the server.

{\mysaveB}
\subsection{Registration}
{\mysaveC}
\label{ssec:registration}


The UAF registration protocol comprises messages between three network entities: a FIDO UAF \textit{client}, a FIDO \textit{server}, and a FIDO \textit{authenticator}.
Registration pairs the client's authenticator with a server for future operations requiring authentication.
As is default in the UAF specification, we assume a first-factor authenticator integrated with the device running the client.
The client and the server communicate over a TLS channel.
The authenticator and the client reside together on a single device and communicate using a custom \textit{Application Programming Interface (API)}~\cite{baghdasaryan2022fido}.
Prior to executing the registration protocol, the client first authenticates to the server using a traditional authentication mechanism, such as providing a username and a password.

The relying party initiates the registration protocol by sending to the client four values: (1)~$username$, the username under which the client will register,
(2)~$policy$, a list of allowed and disallowed authenticator types,
(3)~$appId$, an identifier for the relying party, and
(4)~$challenge$, a freshly generated cryptographic nonce.
The client verifies that $appid$ matches the server's relying party, then proceeds to the next step of the protocol by preparing a registration request.

To prepare the registration request, the client prepares several additional values:
(5)~$facetId$, a set of identifiers for which the registration is valid, included for the purpose of preventing excess registrations for a relying party with multiple valid identities,
(6)~$tlsData$, an optional entry of data for channel binding, and
(7)~$fc$, a final challenge comprising a hash of $appId$, $challenge$, $facetId$, and $tlsData$.
The client transmits $username$, $appId$, and $fc$ to the authenticator via the local API.

The authenticator now generates a registration assertion by generating 
(8)~an asymmetric authentication key pair $(K_{pub}, K_{priv})$ and
(9)~a key handle $h = \hbox{hash}(K_{priv}, username)$.
In addition, the authenticator introduces the following values:
(10)~$aaid$, an identifier for the authenticator's model,
(11)~$cert_{attest}$, a certificate for a pre-loaded attestation key backed by a manufacturer's root attestation certificate,
(12)~\hbox{$reg$-$cntr$} and $cntr$, a pair of operation counters that assist the relying party in detecting duplicate assertion operations, and
(13)~$s$, an assertion signed by the authenticator's attestation key and containing $aaid$, $fc$, \hbox{$reg$-$cntr$}, $cntr$, and $K_{pub}$.
To the client, the authenticator transmits $aaid$, $K_{pub}$, $fc$, $h$, $cert_{attest}$, \hbox{$reg$-$cntr$}, $cntr$, and $s$, which the client forwards to the server over the existing TLS channel.
Figure~\ref{fig:uaf_reg} visualizes the cryptographic flow of the UAF registration network protocol.

{\mysaveB}
\subsection{Authentication}
{\mysaveC}
\label{ssec:authentication}


The UAF authentication protocol is similar to
registration, but  omits some of the values present in registration messages, and the authenticator signs the assertion using $K_{priv}$ (previously generated during registration) rather than with the device's attestation key.
Authentication and registration rely on a TLS channel between the client and the server and apply channel binding in the same manner.
Figure~\ref{fig:uaf_auth} illustrates the message flow of UAF authentication.

When initiating the authentication protocol, the server omits $username$ and transmits only $policy$, $appId$, and $challenge$. The client now transmits only $appId$ (or a similar key identifier for the corresponding relying party) and $fc$ to the authenticator. The authenticator responds with $fc$, a random nonce $n$, $cntr$, and an assertion $s$ consisting of a signature of $fc$, $n$, and $cntr$ signed by $K_{priv}$, the asymmetric authentication key that the authenticator generated during the registration protocol.


{\mysaveA}
\section{Previous Work}
{\mysaveAA}
\label{sec:previous}

In contrast with our work, existing formal-methods studies of FIDO UAF abstract away the TLS channel and neither explicitly model nor analyze the channel bindings.
By neglecting to analyze channel binding, previous work fails to 
examine whether UAF protects against undesirable protocol interactions
caused by the various ways UAF binds channels or chooses not to bind channels.

In 2021, Feng et al.\ \cite{feng2021formal} presented a formal-methods analysis of the UAF registration and authentication protocols using ProVerif.
In their analysis, they formalized and verified UAF authentication goals based on Lowe's authentication hierarchy~\cite{lowe1997hierarchy}.
They discovered and corrected several attacks on the UAF standard, including attacks against enterprise UAF implementations in finance.
Their ProVerif model abstracts away details of the TLS channel, representing it as a communication channel under an identifier.
To bind UAF messages to the channel, the model specifies a function that maps a TLS channel identifier to a unique bitstring representing the channel: a perfect channel binding.
In contrast, the approved UAF bindings bind only to one of the channel's endpoints, may be valid for multiple TLS channels between the same endpoints, and, when using TLS 1.2, depend on the client generating a fresh premaster secret.
Consequently, analysis of their model overlooks potential vulnerabilities resulting from the limitations of UAF's channel bindings.

Using ProVerif, Pereira et al.\ \cite{pereira2018formal} also present a formal-methods analysis of FIDO's U2F protocol, a similar protocol to those in UAF.
They conclude that, if the client verifies the \textit{appId} it receives from the server, the protocol achieves the standard's authentication goals.
Notably, their analysis completely omits channel binding, focusing exclusively on the \textit{appId}'s role in helping the client authenticate the server's challenge: they do not present any analysis of U2F with channel binding.
We consider the \textit{appId} a poor value to bind because it is a public, not a cryptographic value, and, by FIDO's own admission~\cite{baghdasaryan2020fido}, may be possible for an adversary to spoof (e.g., by compromising a subdomain owned by a legitimate organization).

Hu and Zhang~\cite{hu2016security} present an informal analysis of UAF, identifying three attacks: (1)~a misbinding attack, (2)~a parallel-session attack, and (3)~a multi-user attack.
These attacks do not consider the effect of channel binding and 
require a stronger adversary than a DY intruder:
the adversary corrupts the client, corrupts the authenticator, or exploits multiple users sharing an authenticator.


Panos et al.\ \cite{panos2017security} perform an informal analysis of UAF, presenting several high-level attack vectors.
B{\"u}ttner and Gruschka~\cite{buttner2023protecting} present MitM attacks on FIDO extensions, design a protocol to protect the extensions, and analyze the protocol using ProVerif.
Neither of these works considers channel binding.

Additional studies of FIDO UAF assess social-engineering 
attacks~\cite{ulqinaku2021real}, biometric authenticators~\cite{chae2018authentication}, informal trust requirements~\cite{loutfi2015fido}, and feasibility~\cite{chadwick2019improved}.
Other studies analyze related standards such as 
FIDO2~\cite{barbosa2021provable,guan2022formal}, 
for which attacks on UAF are relevant,
though these two studies do not consider channel binding.

Protocols which have been analyzed using CPSA include the CAVES multiparty attestation protocol~\cite{coker2011principles}, Zooko's forced-latency protocol~\cite{lanus2017analysis}, SRP-3~\cite{SRP}, TLS-1.3~\cite{bhandary2021searching}, and the Session Binding Proxy~\cite{SBP}.

Building on our 2021 initial analysis of UAF authentication, 
in 2022, Fuchs, Hammer, and Liu~\cite{fidoreg}
analyzed UAF registration. Our analysis of UAF incorporates and evolves
their analysis. 

{\mysaveA}
\section{Adversarial Model}
{\mysaveAA}
\label{sec:adversarial_model}

We evaluate UAF against a DY-style adversary that seeks to fool legitimate clients into generating assertions on the adversary's behalf, to enable the adversary to register or authenticate in place of an honest user.
To achieve this goal, the adversary transfers challenges from one UAF session to another by exploiting weak or missing cryptographic bindings, including TLS channel bindings.

Our adversary controls all traffic between UAF clients and servers, and is thus capable of adding, dropping, modifying, and replaying arbitrary messages between these entities.
The adversary does not have access to communication between the client and the authenticator, which takes place over a local API rather than over a network channel.
The adversary cannot break standard
cryptographic primitives such as encryption, digital signatures, or hashing,
and to compute these primitives
requires knowledge of correct keys or hash function inputs.
The adversary controls at least one legitimate UAF server on the network, for which they hold appropriate cryptographic keys and certificates.
In practice, this assumption may be an obscure UAF server operating under some dubious or compromised certificate, and may require the adversary to trick an honest user into communicating with it.

We assume that all legitimate network participants communicate with a legitimate certificate authority and deploy the same version of TLS and any channel bindings in use.
Upon request, the certificate authority produces authenticated certificates for any network entity, including the adversary.
Additionally, we assume that any attestation certificate originating from an authenticator is legitimate---in real implementations, the server must check such certificates during UAF registration.
The adversary cannot compromise legitimate authenticators to steal key material.

The adversary is capable of spoofing 
{\it application identifiers (appIds)}---FIDO acknowledges this possible attack in the UAF network protocol specification~\cite{baghdasaryan2020fido}.
In authentication and registration, servers bind their challenges to the \textit{appId}.
Because this value is known to all network participants and is possible to spoof or verify incorrectly, we consider it unsuitable for a cryptographic binding.

Our adversary wishes to access an honest user's accounts to steal information, engage in transactions on the user's behalf, or to usurp the user's identity when interacting with other users.
These activities enable an adversary to sell a user's private data to criminals, blackmail the user with details of their private life, steal the user's assets by transferring them to the adversary, and target the user's social network for further attacks.
Possible targets include a user's online banking,
E-commerce, social media, or business accounts.

{\mysaveA}
\section{CPSA Models (Summary)}
{\mysaveAA}
\label{sec:cpsa-model-abbrev}

For our formal-methods analysis, we created eight CPSA models of FIDO UAF (see Table~\ref{tbl:models}),
including a baseline using passwords over TLS~1.2 without UAF,
six different ways to bind the challenge, and
our suggested improvement that binds the challenge to TLS~1.3 channel key material
derived using TLS-exporter.  For details and sample snippets of our CPSA source code, 
see Sections~\ref{sec:cpsa-model} and \ref{sec:code}.

\begin{table}[!h]
\caption{Summary of our CPSA models. Each UAF model includes the registration and authentication protocols. We include two variations of each TLS 1.2 model: (1)~with RSA key exchange, (2)~with DH key exchange.}
\label{tbl:models}
\begin{center}
\begin{tabular}{l|l}
Name & Summary\\
\hline
Baseline-NoUAF & Plain passwords over TLS 1.2 without UAF.\\
UAF-NoBinding-TLS1.2 & Unbound challenge over TLS 1.2. \\
UAF-TokenBinding-TLS1.2 & Client applies TLS 1.2 token binding to challenge.\\
UAF-ChannelId-TLS1.2 & Client applies TLS 1.2 channel ID binding to challenge. \\
UAF-Endpoint-TLS1.2 & Client binds challenge to hash of server certificate over TLS 1.2. \\
UAF-ServerCert-TLS1.2 & Client binds challenge to server's certificate over TLS 1.2.\\
UAF-NoBinding-TLS1.3 & Unbound challenge over TLS 1.3.\\
UAF-Exporter-TLS1.3 $^*$ & Client binds challenge to TLS 1.3 channel key material.\\
[-4pt]
\end{tabular}
\end{center}
\hspace*{0.25in} ($^*$ Our suggested improvement.)
\end{table}

{\mysaveA}
\section{Security Goals (Summary)}
{\mysaveAA}
\label{sec:goals-abbrev}

FIDO informally states several security goals that it aims to achieve 
via UAF channel binding~\cite{baghdasaryan2022fido},
which we summarize in Table~\ref{tbl:uafgoals} as SG-1, SG-11, SG-12, and SG-13.
In the strand space model, we formalize what we call Goal~1, which implies each of the four 
informally-stated goals.  Specifically, we state Goal~1 as an
\textit{injective agreement}~\cite{lowe1997hierarchy} on a session context between two complementary strands in a UAF strand space (see Section~\ref{sec:goals}).

\begin{table}[!h]   
\caption{Selected UAF security goals as informally stated by the FIDO Alliance~\cite{baghdasaryan2022fido}, for which FIDO specifies channel binding as a security measure.
The goals are vague when referring to notions such as ``high cryptographic strength'' or resilience, and the goals do not refer to specific variables or messages in the protocol.
We formalize these goals and focus on them in our analysis.}
\label{tbl:uafgoals}
\begin{center}
\begin{tabular}{p{0.1\linewidth}  p{0.8\linewidth}}
Goal ID & Description\\
\hline
SG-1 & \textbf{Strong User Authentication:} Authenticate (i.e. recognize) a user and/or a device to a relying party with high (cryptographic) strength.\\[4pt]
SG-11 & \textbf{Forgery Resistance:} Be resilient to Forgery Attacks (Impersonation Attacks). I.e. prevent attackers from attempting to modify intercepted communications in order to masquerade as the legitimate user and login to the system.\\[4pt]
SG-12 & \textbf{Parallel Session Resistance:} Be resilient to Parallel Session Attacks. Without knowing a user’s authentication credential, an attacker can masquerade as the legitimate user by creating a valid authentication message out of some eavesdropped communication between the user and the server.\\[4pt]
SG-13 & \textbf{Forwarding Resistance:} Be resilient to Forwarding and Replay Attacks. Having intercepted previous communications, an attacker can impersonate the legal user to authenticate to the system. The attacker can replay or forward the intercepted messages.\\
\end{tabular}
\end{center}
\end{table}

{\mysaveA}
\section{CPSA Analysis (Summary)}
{\mysaveAA}
\label{sec:analysis-abbrev}

Using CPSA, we analyze for each of our eight models of UAF (see Section~\ref{sec:cpsa-model-abbrev}) whether they achieve Goal~1 (see Section~\ref{sec:goals}).
We do so by
(1)~modeling UAF authentication and registration protocols in the strand space model, 
(2)~for each protocol, defining session contexts (see Section~\ref{ssec:context}), 
(3)~specifying origination assumptions for each of the protocol roles (see Section~\ref{ssec:origination}), 
(4)~using CPSA, search the strand space models to produce ``shapes'' (see Section~\ref{ssecCPSA}), 
(5)~using the shapes, prove Goal~1 (and its constituent sub-goals) true or false.

Each search assumes the perspective of a protocol role and its corresponding origination assumptions
(which might differ for different roles). 
CPSA searches exhaustively for unique shapes (possible protocol executions---see Section~\ref{ssecCPSA}) within a strand space, either proving session-context agreement 
or finding a counterexample (a possible attack).
These searches include intruder strands in strand spaces, where intruders may use
an unbounded number of sessions,
thereby considering all essentially different possible protocol interaction attacks.
For each search, CPSA terminated, 
thereby achieving provable definite results~\cite{liskov2011completeness}.
Table~\ref{tbl:goals} summarizes our findings.

\begin{table}[!h]   
\caption{Summary of our security analysis of our UAF models.
For each model, a check ({\checkmark}) indicates that all strands of the corresponding role (indicated by the column header) satisfy context-agreement theorems with a complementary strand under the role's security assumptions.
A crossmark ({\myno}) indicates that CPSA finds a counterexample to the security goal by
disproving context agreement.
Only a subset of models using TLS1.2-DH and TLS1.3 achieve context agreement from all perspectives: TLS1.2-DH with server side bindings (endpoint, server certificate), and TLS 1.3 (exporter).
These models succeed because they bind to the authenticated server and incorporate mutual freshness in the TLS session via a DH key exchange.}
\label{tbl:goals}
\begin{center}
\begin{tabular}{l|l|l|l|l}
Model & client-reg & server-reg & client-auth & server-auth\\
\hline
Baseline-NoUAF & \checkmark & {\myno} & \checkmark & {\myno}\\[4pt]
UAF-NoBinding-TLS1.2-RSA  & \checkmark & {\myno} & \checkmark & {\myno}\\
UAF-TokenBinding-TLS1.2-RSA & \checkmark & {\myno} & \checkmark & {\myno}\\
UAF-ChannelId-TLS1.2-RSA & \checkmark & {\myno} & \checkmark & {\myno}\\
UAF-Endpoint-TLS1.2-RSA & \checkmark & {\myno} & \checkmark & {\myno}\\
UAF-ServerCert-TLS1.2-RSA & \checkmark & {\myno} & \checkmark & {\myno}\\[4pt]
UAF-NoBinding-TLS1.2-DH  & \checkmark & {\myno} & \checkmark & {\myno}\\
UAF-TokenBinding-TLS1.2-DH & \checkmark & {\myno} & \checkmark & {\myno}\\
UAF-ChannelId-TLS1.2-DH & \checkmark & {\myno} & \checkmark & {\myno}\\
\textbf{UAF-Endpoint-TLS1.2-DH} & \checkmark & \checkmark & \checkmark & \checkmark\\
\textbf{UAF-ServerCert-TLS1.2-DH} & \checkmark & \checkmark & \checkmark & \checkmark\\[4pt]
UAF-NoBinding-TLS1.3 & \checkmark & {\myno} & \checkmark & {\myno}\\
\textbf{UAF-Exporter-TLS1.3 $^*$} & \checkmark & \checkmark & \checkmark & \checkmark\\[-4pt]
\end{tabular}
\end{center}
\hspace*{0.25in} ($^*$ Our suggested improvement.)
\end{table}

\begin{figure}[!b]
\includegraphics[width=\textwidth]{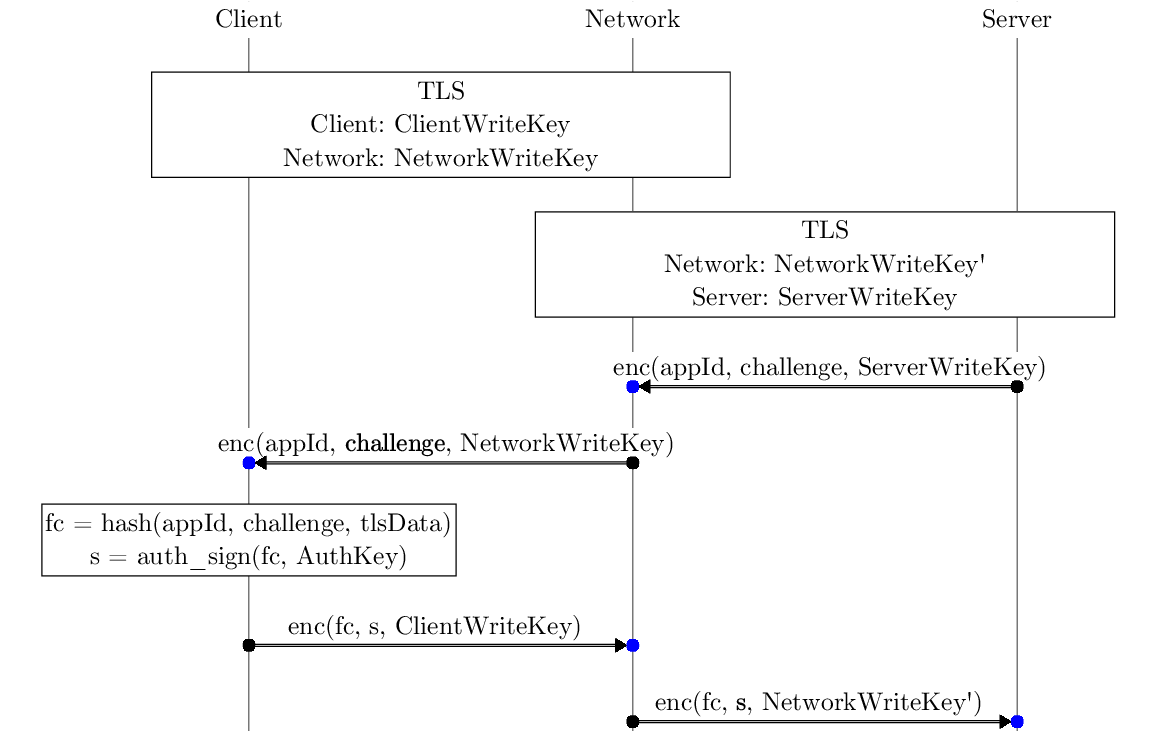}
\caption{
Challenge-reissuing attack during UAF authentication:
CPSA shape for UAF authentication under the server's cryptographic assumptions.
When the client omits channel binding, or includes one of the standard UAF bindings in \textit{tlsData}, the client acts as a confused deputy and produces an attestation for the adversary.
This shape reveals the following structural weakness: in two scenarios, the adversary can masquerade as the client to the server, or transplant messages from one session to another between the same endpoints.
In particular, the adversary can masquerade as the client when the client omits channel binding in the \textit{tlsData}, or when the server fails to verify this binding.
The adversary can transplant messages between such sessions because the client binds only to the endpoints of the channel and not to the session.
}
\label{fig:fido_attack}
\end{figure}


{\mysaveA}
\section{Vulnerabilities, Attacks, and Risks}
{\mysaveAA}
\label{sec:VAR}

We discuss vulnerabilities, attacks, and risks that arise from our analysis of UAF.

{\mysaveB}
\subsection{Vulnerabilities} 
{\mysaveC}
\label{ssecvulnerabilities}

We highlight three vulnerabilities resulting from how UAF binds (or fails to bind) the server's challenge.

\subsubsection{Binding to a TLS1.2-RSA channel is insufficient.}

Because the client is solely responsible for generating the premaster secret in TLS1.2-RSA, the server has no assurances that channel bindings are legitimate.
An adversary that learns or guesses the client's premaster secret can forge such a channel binding to fool the server.
Possible causes include the adversary compromising the client's machine, the client reusing premaster secrets, or an implementation error exposing the premaster secret to the adversary.

\subsubsection{The server cannot verify all channel bindings.}

Depending on which channel bindings a client applies to the final challenge parameters, the server may be unable to verify the bindings when processing an assertion.
For channel bindings relying on client-generated secrets, the server has no assurance that such secrets are not known to the adversary or are unique to the server's TLS session with the client.
When channel binding to the server's identity using TLS1.2-RSA, the server may be unable determine to which TLS channel a client's binding applies.
Additionally, if the client binds to a TLS channel between itself and a perimeter proxy, the server cannot verify the binding.
Under some circumstances, UAF allows a server to continue the protocol even when such bindings fail to verify.
We discuss the issue of such proxies in Section~\ref{ssec:dis-proxy}.

\subsubsection{Channel binding is optional.}

There is no requirement for certified UAF implementations to implement channel binding.
When the client omits a channel binding in the final challenge parameters, the server has no assurance that an adversary does not reissue its challenge in another session.
This scenario is equivalent to when a server fails to verify the channel bindings.

{\mysaveB}
\subsection{Attacks} 
{\mysaveC}
\label{ssecattacks}

Resulting from the weaknesses in UAF's channel binding, an adversary that extracts a challenge from a session of UAF between itself and a legitimate server may reissue this challenge in another session between itself and a legitimate client.
In this \textit{challenge-reissuing attack}, an adversary (acting as a MitM) fools honest clients into signing fraudulent assertions with their authenticators.
UAF registration and authentication feature a similar format for the challenge, assertion, and response, resulting in this attack applying to both protocols.
Figure~\ref{fig:fido_attack} illustrates a 
challenge-reissuing attack during UAF authentication.

Challenge reissuing results from a client's inability to identify for which session the server generates a challenge (see Section~\ref{ssecvulnerabilities}).
Because the server does not cryptographically bind the challenge to any session of UAF, the adversary may forward the challenge to a client who unwittingly uses it to craft an assertion, which the client's authenticator then signs.
UAF's mitigation against such an attack is channel binding, which is optional and may still enable an adversary to move challenges between different sessions involving the same TLS endpoints.
We implement this attack against eBay's implementation of UAF,
which chooses not to use channel binding, 
as allowed by the UAF standard (see Section~\ref{sec:attack}).

An adversary that tricks clients into signing fraudulent assertions can represent themselves as the client to a legitimate UAF server.
In registration, the adversary can register a client's authenticators with a relying party without the client's knowledge.
In authentication, the adversary authenticates as the legitimate client by sending to the server an assertion that the client's authenticator signs.
Without channel binding, the server cannot detect such an assertion.


{\mysaveB}
\subsection{Risks} 
{\mysaveC}
\label{ssecrisks}

We consider the risks posed by challenge reissuing to UAF's primary goal and consider the resulting harmful effects that threaten users and systems authenticating with UAF.

UAF's primary goal, authenticating users and devices to relying parties with a ``high degree of cryptographic strength,'' does not hold when missing or insufficient channel binding enables challenge reissuing.
An adversary can misrepresent themselves as a client and present fraudulent assertions signed by clients unaware that the challenge migrated between UAF sessions.
Even when binding is present, a server may disregard or fail 
to verify the binding correctly, amplifying the risk of this attack.

Because the adversary can fool a client into producing assertions for arbitrary challenges, the adversary can misrepresent themselves as a legitimate user and engage in transactions with the user's authorization.
UAF's design is general purpose, broadly applicable to many systems, and well adopted, enabling an adversary to engage in malicious transactions across a wide range of services (e.g., source control systems, E-commerce websites, online banking, health care, social media) that fail to mitigate challenge reissuing.
Such transactions can have devastating effects on individual users and organizations, which may be victims of identity theft, fraud, blackmail, and other criminal activities.

{\mysaveA}
\section{Attack Implementation}
{\mysaveAA}
\label{sec:attack}

To illustrate the importance of UAF's channel binding, we exploit missing channel binding to carry out a MitM attack against eBay's well-known open-source UAF server~\cite{eBayFIDO}, which is certified by FIDO and implements a subset of UAF 1.0. Our attack enables us to authenticate as an honest user without access to the user's authenticators.


To carry out the protocol, we implement a basic client in the Python programming language. Our client registers dummy authenticators with the server and responds to authentication requests with valid assertions from these authenticators.
We implement a malicious server as a process that passes messages between the client and the UAF server.
Our attack source code is available on GitHub~\cite{PALgithub}.

We assume: (1)~an honest user wishes to communicate with a DY adversary acting as a server; (2)~there is a legitimate server on the network; 
(3)~the adversary controls a subdomain under the server's URL (e.g., the adversary compromises one of the server's trusted facets~\cite{hill2018fidoappid}); and (4)~the client does not implement channel binding.
The adversary wishes to masquerade as the honest client to a legitimate server.

In addition to not supporting channel binding, the eBay FIDO UAF implementation fails to inspect the contents of an assertion's final challenge parameters, verifying only an authenticator's signature.
This egregious omission likely results in additional potential vulnerabilities beyond our work: an adversary can freely substitute challenge parameters without the server's knowledge.
For our attack, we assume only that the server does not enforce any channel binding.

The attack proceeds in two steps---registration and authentication---which exploit lack of binding of the UAF challenge. 

{\mysaveB}
\subsection{Step 1: Registration}
{\mysaveC}

In our attack's first phase, the adversary exploits missing binding to register the honest user's authenticator with the legitimate server, claiming to own this authenticator.

An honest user intentionally attempts to register this authenticator with the adversary, who masquerades as an honest server on the network.
The adversary then engages in two concurrent UAF registration protocols: one between itself and the user, and one between itself and an honest server, presenting the user with the honest server's challenge.
Subsequently, the user issues the challenge to their authenticator, 
which generates a 
{\it key registration data (KRD)} object, including an 
{\it authenticator attestation ID (AAID)}, 
a new public-private key pair, an attestation certificate, a pair of counters (registration and signature), and a hash of these parameters including the malicious challenge.
Using the private key, the authenticator signs the KRD and returns an assertion to the user, who forwards it to the adversary.
The adversary now claims the assertion to the legitimate server, thereby registering the user's authenticator without their knowledge.

{\mysaveB}
\subsection{Step 2: Authentication}
{\mysaveC}

The adversary engages in parallel UAF authentication protocols, transplanting the legitimate server's challenge and policy, which includes the user's authenticator maliciously registered by the adversary, into the session between the adversary and the user.
The user issues the challenge to the authenticator, which builds an authentication assertion,
which the adversary presents to the legitimate server as their own to complete the protocol 
and authenticate.
Claiming the user's identity, the adversary is now free to engage in malicious behaviors under the user's name.

{\mysaveA}
\section{Recommendations}
{\mysaveAA}
\label{sec:rec}

We recommend: (1)~For applications where it is critical to mitigate protocol interactions, the UAF standard requires the client (and potentially the server---see Section~\ref{sec:origin-binding}) to apply a channel binding to the challenge parameters.
(2)~The server should not accept attestations that bind to channels it cannot verify.
(3)~Applications should consider using TLS1.3 and implementing the TLS-exporter binding, which avoids the limitations of TLS channel bindings in UAF by binding to a specific TLS session rather than to one of the communication endpoints.
Applications should not bind to TLS1.2-RSA channels because these may be vulnerable.

{\mysaveA}
\section{Conclusion}
{\mysaveAA}
\label{sec:conclude}

Using CPSA, we performed a formal-methods analysis of channel binding in the FIDO UAF standard. From this analysis, we:
(1)~showed that several of FIDO UAF's channel bindings fail to mitigate protocol interaction, resulting in MitM attacks on the server's challenge.
(2)~implemented an attack against eBay's open-source FIDO UAF server, allowing an adversary to masquerade as a legitimate client.
(3)~recommended and formally verified improvements to TLS channel binding in the UAF standard.

Our attack exploits three limitations of channel binding in FIDO UAF:
(1)~channel binding is optional;
(2)~the server may accept attestations that bind to incorrect channels; and
(3)~UAF binds only to the protocol's \textit{endpoints}, not to the underlying TLS \textit{session}.

Although the UAF specification suggests that omitting channel binding may create a vulnerability (see Section~\ref{sec:fido}), to our knowledge, we are first to carry out a formal-methods analysis of channel binding in UAF and first to exhibit details of an attack on UAF that exploits the weaknesses of UAF's channel binding.
Previous studies of UAF did not analyze the standard's cryptographic channel binding. 
Policy makers should be aware that omitting channel binding, or accepting attestations that bind to an incorrect channel, creates a serious vulnerability in which the adversary can trick the client and authenticator to act as confused deputies to sign an authentication challenge for the adversary.

Despite decades of progress in protocol design, UAF fails to apply cryptographic binding consistently and correctly, enabling potential attacks by a DY adversary.
Our study of UAF channel binding 
illustrates the necessity of adopting formal-methods tools in protocol analysis---including analysis in the design process---and the need to develop and apply rigorous techniques to ensure that critical protocol values are always properly cryptographically bound to their session context.

{\mysaveA}
\section*{Acknowledgments}
{\mysaveAA}

We thank Joshua Guttman for helpful comments and 
valuable guidance on proving security goals using CPSA.
Our work evolves CPSA models of UAF registration created by Danning Liu~\cite{fidoreg}.
Thanks also to Ted Selker for suggesting ways to clarify the presentation.
Alan Sherman and Enis Golaszewski were supported in 2023 in part by the
National Security Agency under an INSuRE+C grant via Northeastern University, and
by the UMBC cybersecurity exploratory grant program.
Alan Sherman was also supported in part by
the National Science Foundation under
DGE grants 1753681 and 2438185 (SFS), 1819521 (SFS Capacity), and 2138921 (SaTC).

\bibliography{binding.bib}

@article{dolevyao,
    author = {Dolev, D. and Yao, A.},
    journal = {IEEE Transactions on Information Theory}, 
    title = {On the Security of Public Key Protocols}, 
    year = {1983},
    volume = {29},
    number = {2},
    pages = {198-208},
    doi = {10.1109/TIT.1983.1056650}
}

@manual{liskov2016cryptographic,
    title = {The Cryptographic Protocol Shapes Analyzer: A Manual},
    author = {Liskov, Moses D and Ramsdell, John D and Guttman, Joshua D and Rowe, Paul D},
    organization = {The MITRE Corporation},
    year = {2016}
}

@incollection{blanchet2013automatic,
    title = {Automatic Verification of Security Protocols in the Symbolic Model: The Verifier {ProVerif}},
    author = {Blanchet, Bruno},
    booktitle = {Foundations of Security Analysis and Design VII},
    pages = {54--87},
    year = {2013},
    publisher = {Springer},
    address = {Berlin, Germany}
}

@inproceedings{meier2013tamarin,
    title={The {TAMARIN} Prover for the Symbolic Analysis of Security Protocols},
    author={Meier, Simon and Schmidt, Benedikt and Cremers, Cas and Basin, David},
    booktitle={International Conference on Computer Aided Verification},
    pages={696--701},
    year={2013},
    publisher={Springer},
    address={Berlin, Germany}
}

@incollection{escobar2009maude,
    title={{Maude-NPA}: Cryptographic Protocol Analysis Modulo Equational Properties},
    author={Escobar, Santiago and Meadows, Catherine and Meseguer, Jos{\'e}},
    booktitle={Foundations of Security Analysis and Design V},
    pages={1--50},
    year={2009},
    publisher={Springer},
    address={Berlin, Germany}
}

@article{millen1987interrogator,
  author={Millen, J.K. and Clark, S.C. and Freedman, S.B.},
  journal={IEEE Transactions on Software Engineering}, 
  title={The {I}nterrogator: Protocol Security Analysis}, 
  year={1987},
  volume={SE-13},
  number={2},
  pages={274-288},
  doi={10.1109/TSE.1987.233151}
}

@inproceedings{cremers2008scyther,
    title={The {S}cyther {T}ool: Verification, Falsification, and Analysis of Security Protocols},
    author={Cremers, Cas JF},
    booktitle={International Conference on Computer Aided Verification},
    pages={414--418},
    year={2008},
    publisher={Springer},
    address={Berlin, Germany}
}

@article{meadows1996nrl,
    title={The {NRL} {P}rotocol {A}nalyzer: An Overview},
    author={Meadows, Catherine},
    journal={The Journal of Logic Programming},
    volume={26},
    number={2},
    pages={113--131},
    year={1996},
    publisher={Elsevier}
}

@article{fabrega1999strand,
    title={Strand Spaces: Proving Security Protocols Correct},
    author={F{\'a}brega, F Javier Thayer and Herzog, Jonathan C and Guttman, Joshua D},
    journal={Journal of Computer Security},
    volume={7},
    number={2/3},
    pages={191--230},
    year={1999}
}

@inproceedings{doghmi2007searching,
  title={Searching for Shapes in Cryptographic Protocols},
  author={Doghmi, Shaddin and Guttman, Joshua and Thayer, Javier},
  booktitle={Tools and Algorithms for the Construction and Analysis of Systems: 13th International Conference},
  pages={523--537},
  year={2007},
  publisher={Springer},
  address={Berlin, Germany}
}

@techreport{liskov2011completeness,
  title={Completeness of {CPSA}},
  author={Liskov, Moses and Rowe, Paul and Thayer, Javier},
  year={2011},
  institution={MITRE},
  note = {\url{https://www.mitre.org/sites/default/files/pdf/12_0038.pdf}}
}

@article{Needham1978,
     author = {Needham, Roger M. and Schroeder, Michael D.},
     title = {{Using Encryption for Authentication in Large Networks of Computers}},
     journal = {Commun. ACM},
     issue_date = {Dec. 1978},
     volume = {21},
     number = {12},
     month = dec,
     year = {1978},
     issn = {0001-0782},
     pages = {993--999},
     numpages = {7},
     url = "http://doi.acm.org/10.1145/359657.359659",
     doi = {10.1145/359657.359659},
     acmid = {359659},
     publisher = {ACM},
     address = {New York, NY, USA},
}

@article{Lowe1995,
	title = {An Attack on the {N}eedham-{S}chroeder Public-Key Authentication Protocol},
	journal = "Information Processing Letters",
	volume = "56",
	number = "3",
	pages = "131--133",
	year = "1995",
	issn = "0020-0190",
	author = "Gavin Lowe",
	note = "\url{http://www.sciencedirect.com/science/article/pii/0020019095001442 }"
}

@inproceedings{feng2021formal,
    title = {A formal analysis of the {FIDO UAF} protocol},
    author = {Feng, Haonan and Li, Hui and Pan, Xuesong and Zhao, Ziming and Cactilab, T},
    booktitle = {Proceedings of 28th Network And Distributed System Security Symposium (NDSS)},
    year = {2021}
}

@inproceedings{pereira2018formal,
  title={Formal Analysis of the {FIDO 1.x} Protocol},
  author={Pereira, Olivier and Rochet, Florentin and Wiedling, Cyrille},
  booktitle={Foundations and Practice of Security: 10th International Symposium, FPS 2017, Nancy, France, October 23-25, 2017, Revised Selected Papers 10},
  pages={68--82},
  year={2018},
  publisher={Springer},
  address={Berlin, Germany}
}

@inproceedings{panos2017security,
    title = {A Security Evaluation of {FIDO’s UAF} Protocol in Mobile and Embedded Devices},
    author = {Panos, Christoforos and Malliaros, Stefanos and Ntantogian, Christoforos and Panou, Angeliki and Xenakis, Christos},
    booktitle = {International Tyrrhenian Workshop on Digital Communication},
    pages = {127--142},
    year= {2017},
    publisher = {Springer},
    address = {Berlin, Germany}
}

@article{hu2016security,
    title = {Security Analysis of an Attractive Online Authentication Standard: {FIDO UAF} Protocol},
    author = {Hu, Kexin and Zhang, Zhenfeng},
    journal = {China Communications},
    volume = {13},
    number = {12},
    pages = {189--198},
    year = {2016},
    publisher = {IEEE}
}

@inproceedings{ulqinaku2021real,
    author = {Enis Ulqinaku and Hala Assal and AbdelRahman Abdou and Sonia Chiasson and Srdjan Capkun},
    title = {Is Real-time Phishing Eliminated with {FIDO}? Social Engineering Downgrade Attacks against {FIDO} Protocols},
    booktitle = {30th USENIX Security Symposium (USENIX Security 21)},
    year = {2021},
    isbn = {978-1-939133-24-3},
    pages = {3811--3828},
    url = {https://www.usenix.org/conference/usenixsecurity21/presentation/ulqinaku},
    publisher = {USENIX Association},
    month = aug,
    address = {Berkley, CA}
}

@article{chae2018authentication,
  title={Authentication Method using Multiple Biometric Information in {FIDO} Environment},
  author={Chae, Cheol-Joo and Cho, Han-Jin and Jung, Hyun Mi},
  journal={Journal of Digital Convergence},
  volume={16},
  number={1},
  pages={159--164},
  year={2018},
  publisher={The Society of Digital Policy and Management}
}

@inproceedings{loutfi2015fido,
  title={{FIDO} Trust Requirements},
  author={Loutfi, Ijlal and J{\o}sang, Audun},
  booktitle={Nordic Conference on Secure IT Systems},
  pages={139--155},
  year={2015},
  publisher={Springer},
  address={Berlin, Germany}
}

@article{chadwick2019improved,
  title={Improved Identity Management with Verifiable Credentials and {FIDO}},
  author={Chadwick, David W and Laborde, Romain and Oglaza, Arnaud and Venant, R{\'e}mi and Wazan, Samer and Nijjar, Manreet},
  journal={IEEE Communications Standards Magazine},
  volume={3},
  number={4},
  pages={14--20},
  year={2019},
  publisher={IEEE}
}

@article{abadi1996prudent,
    title = {Prudent Engineering Practice for Cryptographic Protocols},
    author = {Abadi, Martin and Needham, Roger},
    journal = {IEEE Transactions on Software Engineering},
    volume = {22},
    number = {1},
    pages = {6--15},
    year = {1996},
    publisher = {IEEE}
}

@inproceedings{kelsey1997protocol,
    title = {Protocol Interactions and the Chosen Protocol Attack},
    author = {Kelsey, John and Schneier, Bruce and Wagner, David},
    booktitle = {International Workshop on Security Protocols},
    pages = {91--104},
    year = {1997},
    publisher = {Springer},
    address = {Berlin, Germany}
}

@misc{rfc5929,
	series = {Request for Comments},
	number = 5929,
	howpublished = {RFC 5929},
	publisher =	{RFC Editor},
	doi = {10.17487/RFC5929},
	url = {https://rfc-editor.org/rfc/rfc5929.txt},
    author = {Jeffrey E. Altman and Larry Zhu and Nicolás Williams},
	title =	{{Channel Bindings for TLS}},
	pagetotal =	15,
	year = 2010,
	month =	jul,
}

@misc{rfc8471,
	series = {Request for Comments},
	number = {8471},
	howpublished = {RFC 8471},
	publisher =	{RFC Editor},
	doi = {10.17487/RFC8471},
	url = {https://rfc-editor.org/rfc/rfc8471.txt},
    author = {Andrei Popov and Magnus Nystrom and Dirk Balfanz and Adam Langley and Jeff Hodges},
	title =	{{The Token Binding Protocol Version 1.0}},
	pagetotal =	{18},
	year = {2018},
	month =	{Oct}
}

@misc{balfanz2013channel,
    number =    {draft-balfanz-tls-channelid-01},
    type =      {Internet-Draft},
    institution =  {Internet Engineering Task Force},
    publisher = {Internet Engineering Task Force},
    note =      {Work in Progress},
    url =       {https://datatracker.ietf.org/doc/draft-balfanz-tls-channelid/01/},
    author =    {Dirk Balfanz and Ryan Hamilton},
    title =     {{Transport Layer Security (TLS) Channel IDs}},
    pagetotal = 18,
    year =      2013,
    month =     jun,
    day =       29,
}

@misc{technote2018,
    author = {Balfanz, Dirk},
    title = {{FIDO} {T}ech{N}ote: The Growing Role of Token Binding},
    howpublished = {\url{https://web.archive.org/web/20230204093428/www.fidoalliance.org/fido-technotes-channel-binding-and-fido/}},
    note = {Accessed: 2023-4-8},
    year = {2016}
}

@techreport{baghdasaryan2020fido,
    title ={{FIDO UAF Protocol Specification v1.2}},
    author = {Baghdasaryan, Davit and Balfanz, Dirk and Hill, Brad and Hodges, Jeff and Yang, Ka},
    year = {2020},
    institution = {{FIDO Alliance}},
    howpublished = {\url{https://bit.ly/41h2QDE}}
}

@misc{fidomembers,
     title={{FIDO} Alliance Member Companies and Organizations},
     url={https://fidoalliance.org/members/},
     journal={{FIDO} Alliance},
     author={{FIDO} Alliance},
     year={2023},
     month={Jan}
 }

@misc{fidocert,
    title = {{FIDO} Certified Products},
    url = {https://fidoalliance.org/certification/fido-certified-products/},
    journal={{FIDO} Alliance},
    author={{FIDO} Alliance},
    year={2023},
    month={Jan}
 }

@techreport{hill2018fidoappid,
    title = {{FIDO AppID} and Facet Specification},
    author = {Hill, Brad and Balfanz, Dirk and Baghdasaryan, Davit},
    year = {2018},
    institution = {{FIDO Alliance}},
    howpublished = {\url{https://bit.ly/2VYVNhi}}
 }

@misc{fido2,
    title ={Web {A}uthentication: An {API} for accessing Public Key Credentials {L}evel 2},
    author = {Hodges, Jeff and Jones, J. C. and Jones, Michael B. and Kumar, Akshay and Lundberg, Emil},
    year = {2021},
    institution = {World Wide Web Consortium},
    howpublished = {\url{https://www.w3.org/TR/webauthn/}}
}

@article{hardy1988confused,
  title={The Confused Deputy: (Or Why Capabilities Might Have Been Invented)},
  author={Hardy, Norm},
  journal={ACM SIGOPS Operating Systems Review},
  volume={22},
  number={4},
  pages={36--38},
  year={1988},
  publisher={ACM},
  address= {New York, NY}
}

@inproceedings{paulson1999inductive,
  title={Inductive Analysis of the Internet Protocol {TLS}: Transcript of Discussion},
  author={Paulson, Larry},
  booktitle={Security Protocols: 6th International Workshop Cambridge, UK, April 15--17, 1998 Proceedings 6},
  pages={13--23},
  year={1999},
  publisher={Springer},
  address={Berlin, Germany}
}

@inproceedings{buttner2023protecting,
  title={Protecting {FIDO} Extensions Against Man-in-the-Middle Attacks},
  author={B{\"u}ttner, Andre and Gruschka, Nils},
  booktitle={Emerging Technologies for Authorization and Authentication: 5th International Workshop, ETAA 2022, Copenhagen, Denmark, September 30, 2022, Revised Selected Papers},
  pages={70--87},
  year={2023},
  publisher={Springer},
  address = {Berlin, Germany}
}

@inproceedings{o2016tls,
    title={{TLS} proxies: Friend or foe?},
    author={O'Neill, Mark and Ruoti, Scott and Seamons, Kent and Zappala, Daniel},
    booktitle={Proceedings of the 2016 Internet Measurement Conference},
    pages={551--557},
    year={2016},
    publisher={ACM},
    address={New York, NY}
}

@incollection{SRP,
  title={Formal methods analysis of the {S}ecure {R}emote {P}assword {P}rotocol},
  author={Sherman, Alan T and Lanus, Erin and Liskov, Moses and Zieglar, Edward and Chang, Richard and Golaszewski, Enis and Wnuk-Fink, Ryan and Bonyadi, Cyrus J and Yaksetig, Mario and Blumenfeld, Ian},
  booktitle={Logic, Language, and Security},
  pages={103--126},
  year={2020},
  publisher={Springer}
}

@unpublished{FIDO,
    title = {Cryptographic Binding Should Not Be Optional:
    A Formal-Methods Analysis of {FIDO} {UAF} Authentication},
    author = {Golaszewski, Enis and Alan T. Sherman and Edward Zieglar and Jonathan D. Fuchs and Sophia Hamer},
    howpublished = {CSEE Dept, UMBC, unpublished manuscrip},
    date = {2022}
}

@misc{fidoreg,
    title = {A Man-in-the-Middle Attack on the {FIDO UAF} Registration Protocol},
    author = {Fuchs, Jonathan and Hamer, Sophia and Liu, Danning},
    howpublished = {CMSC-691 Special Topics: Cybersecurity Research (INSuRE) course project, CSEE Dept, UMBC, unpublished manuscript},
    year = {2022}
}

@misc{eBayFIDO,
    author = {eBay},
    year = {2022},
    title = {{eBay FIDO UAF Implementation}},
    howpublished = {\url{https://github.com/eBay/UAF}},
    note = {Accessed: 2023-03-31}
}

@misc{SBP,
  title={Limitations of Wrapping Protocols and {TLS} Channel Bindings: Formal-Methods Analysis of the {S}ession {B}inding {P}roxy Protocol},
  author={Golaszewski, Enis and Zieglar, Edward and Sherman, Alan T and Elsaad, Kirellos Abou and Fuchs, Jonathan},
  year = {2024},
  month = {March},
  howpublished = {unpublished manuscript,  CSEE Dept, UMBC}
}

@misc{PALgithub,
    author = {UMBC Protocol Analysis Lab},
    title = {{PAL} {GitHub} Repository},
    month = {April},
    year = {2023},
    publisher = {GitHub},
    journal = {GitHub repository},
    howpublished = {\url{https://tinyurl.com/3d2wnhuf}},
}

@misc{baghdasaryan2022fido,
  title={{FIDO UAF Authenticator-Specific Module API}},
  author={Baghdasaryan, Davit and Hill, Brad and Sasson, Roni and Hodges, Jeff and Yang, Ka},
  year={2022},
  publisher={Implementation Draft. URL: https://fidoalliance. org/specs/fido-uaf-v1}
}

@inproceedings{lowe1997hierarchy,
  title={A hierarchy of authentication specifications},
  author={Lowe, Gavin},
  booktitle={Proceedings 10th computer security foundations workshop},
  pages={31--43},
  year={1997},
  organization={IEEE}
}

@misc{rfc9266,
    series =    {Request for Comments},
    number =    9266,
    howpublished =  {RFC 9266},
    publisher = {RFC Editor},
    doi =       {10.17487/RFC9266},
    url =       {https://www.rfc-editor.org/info/rfc9266},
    author =    {Sam Whited},
    title =     {{Channel Bindings for TLS 1.3}},
    pagetotal = 7,
    year =      2022,
    month =     jul,
}

@book{rescorla2001ssl,
  title={SSL and TLS: Designing and building secure systems},
  author={Rescorla, Eric},
  publisher = {Addison-Wesley},
  year={2001}
}

@techreport{ietf-tokbind-ttrp-08,
    number =    {draft-ietf-tokbind-ttrp-08},
    type =      {Internet-Draft},
    institution =   {Internet Engineering Task Force},
    publisher = {Internet Engineering Task Force},
    note =      {Work in Progress},
    url =       {https://datatracker.ietf.org/doc/draft-ietf-tokbind-ttrp/08/},
    author =    {Brian Campbell},
    title =     {{HTTPS Token Binding with TLS Terminating Reverse Proxies}},
    pagetotal = 14,
    year =      2019,
    month =     4,
    day =       15,
}

@inproceedings{guan2022formal,
  title={A formal analysis of the {FIDO2} protocols},
  author={Guan, Jingjing and Li, Hui and Ye, Haisong and Zhao, Ziming},
  booktitle={European Symposium on Research in Computer Security},
  pages={3--21},
  year={2022},
  organization={Springer}
}

@article{coker2011principles,
  author       = {George Coker and
                  Joshua D. Guttman and
                  Peter A. Loscocco and
                  Amy L. Herzog and
                  Jonathan K. Millen and
                  Brian O'Hanlon and
                  John D. Ramsdell and
                  Ariel Segall and
                  Justin Sheehy and
                  Brian T. Sniffen},
  title        = {Principles of remote attestation},
  journal      = {Int. J. Inf. Sec.},
  volume       = {10},
  number       = {2},
  pages        = {63--81},
  year         = {2011},
  url          = {https://doi.org/10.1007/s10207-011-0124-7},
  doi          = {10.1007/S10207-011-0124-7},
  bibsource    = {dblp computer science bibliography, https://dblp.org}
}

@inproceedings{bhandary2021searching,
  author       = {Prajna Bhandary and
                  Edward Zieglar and
                  Charles Nicholas},
  editor       = {Daniel Dougherty and
                  Jos{\'{e}} Meseguer and
                  Sebastian Alexander M{\"{o}}dersheim and
                  Paul D. Rowe},
  title        = {Searching for Selfie in {TLS} 1.3 with the {C}ryptographic {P}rotocol
                  {S}hapes {A}nalyzer},
  booktitle    = {Protocols, Strands, and Logic - Essays Dedicated to Joshua Guttman
                  on the Occasion of his 66th Birthday},
  series       = {Lecture Notes in Computer Science},
  volume       = {13066},
  pages        = {50--76},
  publisher    = {Springer},
  year         = {2021},
  url          = {https://doi.org/10.1007/978-3-030-91631-2\_3},
  doi          = {10.1007/978-3-030-91631-2\_3},
  bibsource    = {dblp computer science bibliography, https://dblp.org}
}

@article{lanus2017analysis,
  title={Analysis of a forced-latency defense against man-in-the-middle attacks},
  author={Lanus, EF and Zieglar, EV},
  journal={Journal of Information Warfare},
  volume={16},
  number={2},
  pages={66--78},
  year={2017},
  publisher={JSTOR}
}

@inproceedings{cremers2017comprehensive,
  title={A comprehensive symbolic analysis of {TLS} 1.3},
  author={Cremers, Cas and Horvat, Marko and Hoyland, Jonathan and Scott, Sam and van der Merwe, Thyla},
  booktitle={Proceedings of the 2017 ACM SIGSAC conference on computer and communications security},
  pages={1773--1788},
  year={2017}
}

@article{tran2022formal,
    title={Formal verification of {TLS 1.2} by automatically generating proof scores},
    author={Tran, Duong Dinh and Ogata, Kazuhiro},
    journal={Computers \& Security},
    volume={123},
    pages={102909},
    year={2022},
    publisher={Elsevier}
}

@misc{lodderstedt2025rfc,
  title={{RFC} 9700: {Best} Current Practice for {OAuth} 2.0 Security},
  author={Lodderstedt, T and Bradley, J and Labunets, A and Fett, D},
  year={2025},
  publisher={RFC Editor}
}

@inproceedings{barbosa2021provable,
  title={Provable security analysis of {FIDO2}},
  author={Barbosa, Manuel and Boldyreva, Alexandra and Chen, Shan and Warinschi, Bogdan},
  booktitle={Annual International Cryptology Conference},
  pages={125--156},
  year={2021},
  organization={Springer}
}
\bibliographystyle{splncs04}

\vfill
\noindent 
({\today}) A shorter version of this paper appears in the \textit{Proceedings of Security Standardisation Research (SSR) 2025}, published by Springer in the LNCS series.

\clearpage
\appendix

\section{Acronyms and Abbreviations}
\label{sec:abbrevs}


\noindent \begin{tabular}{ll} 
AAID & Authenticator Attestation ID\\
API & Application Programming Interface\\
appID & Application Identifier\\
ASM & Authenticator-Specific Module\\
CA & Certificate Authority\\
CPSA    & Cryptographic Protocol Shapes Analyzer\\
DH     & Diffie-Hellman\\
DY   & Dolev-Yao\\
FIDO   & Fast Identity Online\\
KRD & Key Registration Data\\
MitM & man-in-the-middle\\
NS   & Needham-Schroeder\\
RSA & Rivest Shamir Adleman\\
TLS  & Transport Layer Security\\
UAF   & Universal Authentication Framework 
\end{tabular}

\section{A Structural Weakness of FIDO UAF}
{\mysaveAA}
\label{sec:weakness}

FIDO UAF has a fundamental structural flaw: with, or without channel binding (see Section~\ref{ssec:binding}),
the client cannot verify if the server-generated challenge originated from
the current session between the client and server. This flaw 
results from inadequate binding of the challenge to the session context.

As a result, in certain situations, a DY adversary is able to carry out
significant malicious actions by manipulating and transplanting 
parts of messages (e.g., containing the challenge) between different
protocol sessions.  We summarize three examples.

{\it Example 1}. When channel binding is not used, the adversary masquerades as
a legitimate bank server.  When the client initiates a session with the adversary,
the adversary launches a parallel session with the legitimate server.
The adversary obtains a challenge from the legitimate server and sends
it to the client.  Unable to verify the context of the challenge, the client
returns a signed attestation of the challenge to the adversary, who
passes it along to the server, thereby authenticating to the server as
the client.  In Section~\ref{sec:attack} we implement this MitM attack against 
eBay's open-source FIDO UAF server.

{\it Example 2}. Many organizations install a perimeter TLS proxy for the purpose
of monitoring traffic flows across their perimeters~\cite{o2016tls}.
To support this practice, when channel binding is used, FIDO UAF permits the server to accept channel bindings to the proxy rather than to the server.  
In the DY model, however, the
proxy might be malicious. Because the client cannot verify the source of the
challenge, the client cannot distinguish whether the challenge is from
a session with the server, a legitimate proxy, or a malicious proxy.
In this sense, FIDO UAF supports an adversary to carry out a MitM attack
between the client and the server.

{\it Example 3}. When channel binding is used with or without perimeter TLS proxies, a potential subtle vulnerability arises when the client carries out multiple sessions with the same server.  Because ``channel binding'' in UAF
binds only to the \textit{endpoints} of a channel, and not to the \textit{session}, 
potential threats arise in which the adversary manipulates messages
among the multiple sessions.  For example, suppose a client establishes
two concurrent sessions with an investment bank for the purpose of 
making a stock transaction in each session.  Although the server generates
a different challenge for each session, the adversary might be able to manipulate
the challenges, and the client's signed attestations of them,
to change the order of the transactions. Changing the order of transactions
can have significant consequences.  Carrying out this attack would require
dealing with other complexities, including the authenticator's signature counter,
if the authenticator has a signature counter.  Even if the authenticator has
a signature counter, it might be possible for the adversary to send the
client a policy that states that the adversary will accept only authenticators
that do not use signature counters.

What UAF does is for the server to send a challenge to the client,
the client to return a signed assertion of this challenge, and
the server to verify the signature and the contents of the assertion, including
the channel binding if present.  These actions, however, have serious
limitations.  The client cannot verify from what session the challenge
originates; the client can choose not to use channel
binding; the server may accept incorrect bindings; and channel
binding binds only the endpoints of the channel and not the session.
These limitations exist whether the authenticator is based on passwords
or biometrics.

UAF's failure to bind challenges at the server adequately suggests a need for formal verification of each component during the design process of a complex framework like UAF.
We discuss this necessity of formal verification in Section~\ref{ssec:formal-methods}.

\section{Background}
{\mysaveAA}
\label{sec:background}

We present brief background for cryptographic binding, formal-methods analysis of protocols, CPSA, and TLS, which we use in subsequent sections to perform a formal-methods analysis of FIDO UAF authentication.

{\mysaveB}
\subsection{Cryptographic Binding}
{\mysaveC}
\label{ssec:cryptobinding}

In 1996, Abadi and Needham~\cite{abadi1996prudent} presented informal guidelines for designing sound cryptographic protocols, including the need for explicit \textit{context} and \textit{cryptographic binding}.
Cryptographic binding associates sensitive data with a specific context, complicating the malicious act of transplanting data from one protocol context to another.
Digital certificates are a well-known example of binding: the certificate associates an entity's identifier with the entity's public key, bound by the digital signature of a trusted issuer.

Network protocols that fail to bind messages to a context are vulnerable to \textit{protocol interaction}~\cite{kelsey1997protocol}, in which an adversary uses messages between different protocols, or between different sessions of the same protocol, 
to produce undesirable outcomes.
Often, an adversary will produce interactions between two sessions of the same protocol, 
as in a MitM attack.
For example, Gavin Lowe's~\cite{Lowe1995} 1995 famous attack on the NS public-key protocol~\cite{Needham1978} exploits a lack of binding between a random nonce and its owner, enabling an adversary to misrepresent another communicant's nonce as their own.
A protocol that fails to bind messages may also permit an adversary to transplant data or entire messages to attack a 
separate protocol.
Protocol interaction can be mitigated by
binding cryptographic values to a specific protocol session.

{\mysaveB}
\subsection{Formal Methods for Protocol Analysis}
{\mysaveC}
\label{ssec:background-formal-methods}

Current tools for formal-methods analysis of cryptographic protocols include\linebreak 
ProVerif~\cite{blanchet2013automatic}, 
Tamarin Prover~\cite{meier2013tamarin}, 
Maude-NPA~\cite{escobar2009maude}, and 
CPSA~\cite{liskov2016cryptographic}.
These tools build on ideas of legacy tools, including NPA~\cite{meadows1996nrl}, Interrogator~\cite{millen1987interrogator}, and Scyther~\cite{cremers2008scyther}, and efforts such as Lowe's analysis of NS using FDR, a refinement checker for formal {communicating sequential processes} models~\cite{Lowe1995}.

ProVerif proves properties of Horn clauses modeling a protocol.
Tamarin Prover proves properties about multiset rewriting rules.
Maude-NPA searches backwards from attack states through unification.
By contrast, CPSA searches for protocol interactions, and upon termination, provably enumerates all essentially different protocol interactions for a protocol model~\cite{doghmi2007searching,liskov2011completeness}.

It is possible to carry out inductive security analysis of protocols using high-order logic theorem provers such as Isabelle~\cite{paulson1999inductive}.
This approach uses general interactive theorem provers, rather than specialized tools for protocol analysis.

As a ``model-finding tool,'' CPSA is capable of discovering protocol security goals given a protocol specification, a set of assumptions, and a partial execution.
As do other tools, CPSA searches for attacks over an unbounded number of parallel sessions.
Additionally, CPSA is also capable of verifying stated theorems within the model. 
For our analysis, we use CPSA because of our familiarity with the tool, our proximity to CPSA experts, and the tool's model-finding properties, which are useful for identifying protocol security goals and protocol interactions.

{\mysaveB}
\subsection{Strand Spaces} 
{\mysaveC}
\label{ssecstrands}

To enable analysis with CPSA, we formalize UAF protocols as strand spaces, which comprise protocol actors called \textit{strands}~\cite{fabrega1999strand}.
Strands are sequences of incoming and outgoing messages on a Dolev-Yao network.
Each message contains \textit{terms}, which represent message primitives such as identifiers, nonces, or encrypted information.
Causal relationships between messages that strands emit and receive form \textit{bundles}, which may express correct protocol executions between honest strands or corrupted executions involving \textit{penetrator} strands that model behaviors of a DY adversary.
Bundles must satisfy four requirements:
(1)~A bundle must be finite.
(2)~If a strand in a bundle receives a message, that message must be sent by some other unique strand in the bundle.
(3)~If some term in a strand is included in the bundle, any term that precedes it must also be present; and
(4)~the graph resulting from the bundle (using term sequences and message relationships as edges) must be acyclic.

From their specifications, we model protocols as strands and use CPSA to enumerate all essentially different bundles that include these strands.
Bundles are valuable tools for investigating attacks on protocols because they can describe corrupt protocol executions in which penetrator strands manipulate network traffic.
Penetrator strands can send and receive arbitrary messages, concatenate and separate message terms to synthesize new messages, and apply knowledge of secret keys to encrypt or decrypt information.
Additionally, penetrator strands can introduce new keys, identities, and information into the network to masquerade as honest network entities.
By inspecting bundles that express corrupt executions, the behaviors of penetrator strands provide clues that enable potential attacks against a protocol's security goals.

{\mysaveB}
\subsection{CPSA} 
{\mysaveC}
\label{ssecCPSA}

CPSA~\cite{liskov2016cryptographic}
is an open-source tool that analyzes cryptographic network protocols for protocol interactions.
In contrast to some formal-methods tools, CPSA is not merely a theorem-prover but a model finder.
For an input model, which comprises roles, messages, variables, and assumptions regarding those variables, CPSA outputs trees of graphical \textit{shapes} that illustrate possible protocol executions resulting from an unbounded number of sessions.
When CPSA terminates, it provably discovers all essentially different shapes 
for the input model~\cite{liskov2011completeness}, 
enabling users to inspect these shapes and identify all protocol interactions possible for the input model.

Users define CPSA models using LISP-like s-expressions that implement a custom language.
In these models, which superficially resemble (but are not) executable source code, users specify one or more roles, associated variables and messages, and \textit{skeletons}.
Skeletons specify one or more initial roles and impose assumptions on the role variables, such as forcing a value to originate uniquely in each session or being unavailable to the adversary.
When CPSA executes, it attempts to satisfy skeletons into shapes by repeatedly applying authentication tests in {strand space theory}~\cite{fabrega1999strand}.

Shapes, which CPSA represents graphically, consist of \textit{strands} that each represent a legitimate protocol role or an adversary's listener for key values.
Each strand consists of sequential \textit{nodes} that specify message transmission or reception events.
Connecting these nodes between different strands are two types of arrows: 
\textit{solid arrows} indicate a pair of message events, transmission and reception, 
for which CPSA can prove a causal relationship.
\textit{Dashed arrows} indicate a pair of message events for which CPSA, 
acting as a DY adversary, manipulates available information from the send event to satisfy the receive event.
When analyzing CPSA shapes, users take special note of dashed arrows 
because dashed arrows often suggest undesirable protocol interactions.

A user may specify security goals comprising predicates, 
which all output shapes must satisfy.
Specifying such security goals enables the user to prove authentication 
and other security properties of the protocol in a manner similar to that for other tools.

Section~\ref{ssec:method} outlines the workflow involving CPSA that we used to analyze FIDO UAF, which includes stating and proving (or disproving) security goals.
When inspecting CPSA output, we encourage readers to inspect our graphical ``.xhtml'' shapes files~\cite{PALgithub}.

{\mysaveB}
\subsection{Transport Layer Security (TLS)} 
{\mysaveC}
\label{ssecTLS}

{\it Transport Layer Security (TLS)}~\cite{rescorla2001ssl} is one of the most widely-deployed cryptographic protocols in the world, and is used by millions of systems to negotiate encrypted communication channels on insecure networks.
The most common application of TLS is to secure communication between a user and 
web server, enabling the user to communicate with the server with some guarantee of confidentiality and integrity.
To authenticate itself to clients, the server presents a certificate, signed by a certificate authority, that attests a public key associated with the server.
Because most clients do not possess or present client certificates, TLS often provides only one-way authentication, requiring the server to authenticate the client using an additional protocol that executes over the TLS channel.
In a mutual authentication variation of TLS known as \textit{mutual TLS}, or \textit{mTLS}, the client produces a certificate that enables the server to authenticate the client as part of the TLS handshake.

Protocols such as UAF assume execution over a pre-negotiated TLS channel, and may include provisions to bind messages and data to the channel using \textit{channel bindings}.
We give an overview of UAF's channel bindings in Section~\ref{ssec:binding}.

In our analysis, we consider both TLS 1.2 and TLS 1.3, which are actively used and supported in 2025.
For TLS 1.2, we consider two types of key exchange: RSA and DH.
In TLS 1.2 with RSA, the \textit{channel master secret} derives \textit{freshness}---uniqueness to a single protocol session---from three values: a server random nonce, a client random nonce, and a client premaster secret.
Under this scenario, the freshness of the master secret depends crucially on the freshness of the client's premaster secret.
Because the client may use a premaster secret that is not fresh to the session, an adversary may hijack a TLS 1.2 session when using RSA.
TLS 1.2 with DH, and TLS 1.3, improve the freshness guarantee of the master secret by replacing the RSA premaster secret with a DH key exchange, ensuring that both the client and the server can contribute freshness to a secret.


\clearpage
\section{Ladder Diagrams}
{\mysaveAA}
\label{sec:ladder}

\begin{figure}[!h]
\includegraphics[width=\textwidth]{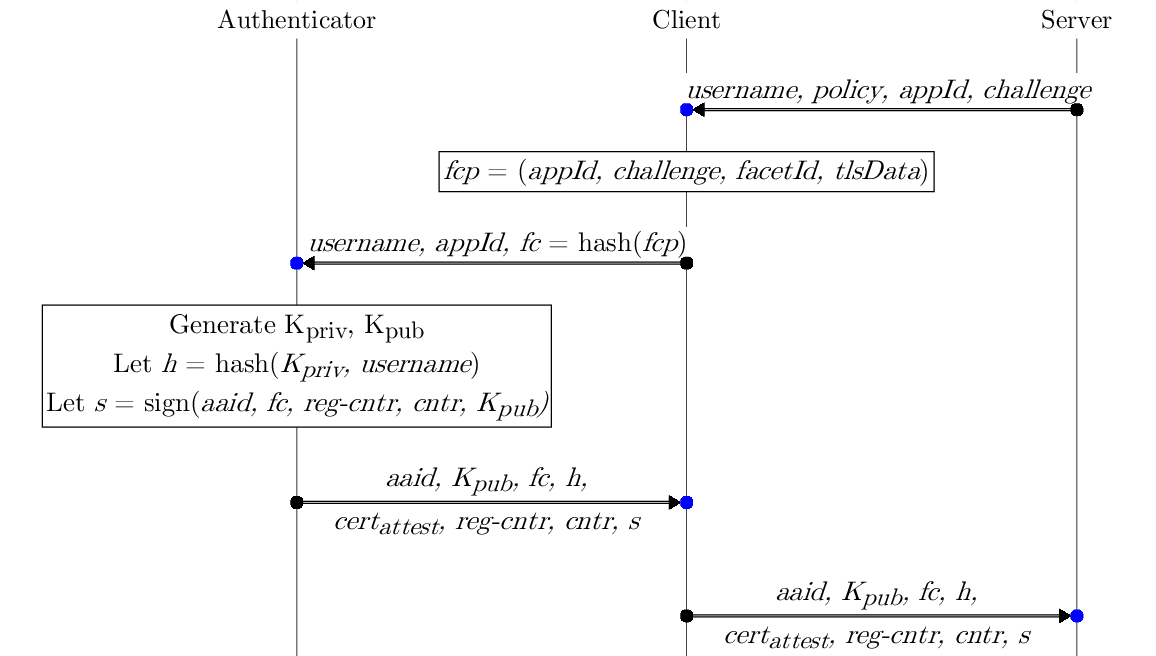}
\caption{
Idealized message sequence diagram of UAF registration over a network. Vertical lines correspond to protocol roles; horizontal arrows indicate message flows from one role to another; and 
arrow labels indicate message contents. Information in boxes specify local processing for each role. The client and server communicate over a pre-negotiated TLS channel. The client communicates with the authenticator via the \textit{Authenticator-Specific Module (ASM)} API. The authenticator signs $s$ using the attestation key corresponding to the attestation certificate that the manufacturer loads on the device.
}
\label{fig:uaf_reg}
\end{figure}

\clearpage
\begin{figure}[h]
\includegraphics[width=\textwidth]{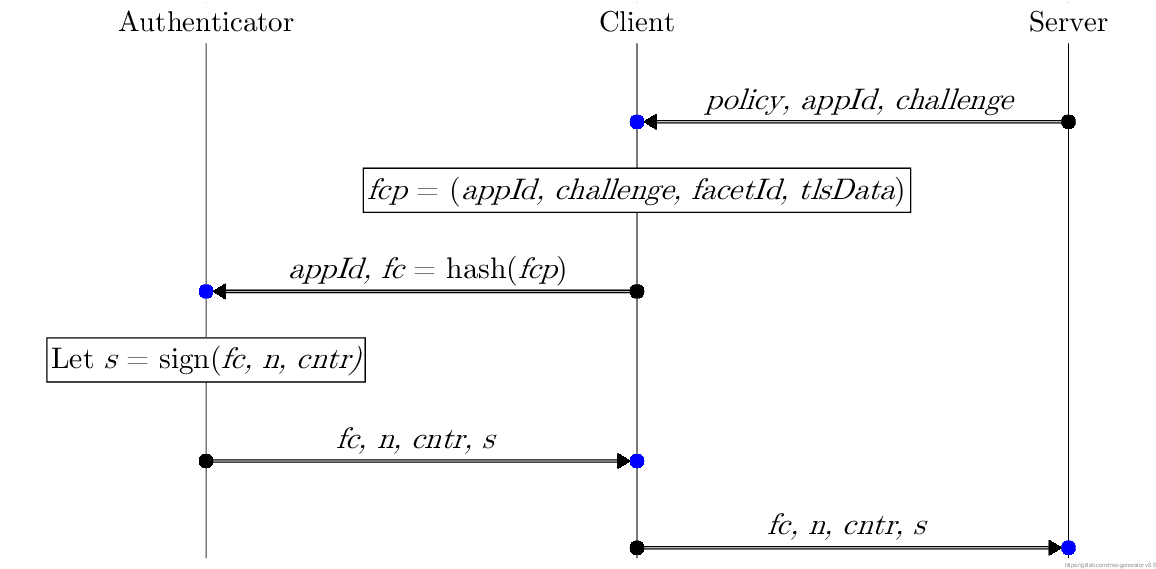}
\caption{
Idealized message sequence diagram of UAF authentication over a network. 
As in registration, the client and server communicate over a pre-negotiated TLS channel, and the client communicates with the authenticator using a local API. The authenticator signs $s$ using $K_{priv}$, which it recalls from a successful run of the registration protocol.
}
\label{fig:uaf_auth}
\end{figure}


\section{Models}
{\mysaveAA}
\label{sec:cpsa-model}

Using CPSA, we model variations of UAF that differ in their version of TLS, including TLS 1.2 with RSA and TLS 1.2 with DH key exchange, and the method by which they bind to the TLS channel.
Additionally, we model a baseline scenario that performs password-based authentication without UAF.
Table~\ref{tbl:models} summarizes these models.

In each UAF model, we specify client and server strands for the registration and authentication network protocols, embedding the authenticator into the client protocol role.
The UAF model names reflect the TLS version and the channel binding they incorporate, including a TLS~1.3 exporter binding that the UAF standard does not support (see Section~\ref{ssec:binding}).
Because TLS~1.3 does not officially support TLS~1.2 bindings, we omit such models.

Each of our models incorporates our CPSA models of TLS~1.2 (with RSA and DH key exchange) or TLS~1.3, 
which reflect the common use case in which a client authenticates a server via the server's certificate (these models are available in~\cite{PALgithub}).
Using CPSA, we verified the secrecy and authentication properties of our TLS~1.2 and 1.3 models, 
and found that our models are consistent with previous formal analyses of TLS~\cite{cremers2017comprehensive,tran2022formal}. 
As we explain in Section~\ref{sec:analysis}, some of the security limitations of UAF stem from
limitations of TLS (e.g., freshness of the premaster secret when using RSA).

Each UAF model comprises four strands: \textit{client-reg} and \textit{server-reg} strands to model registration, and \textit{client-auth} and \textit{server-auth} strands to model authentication.
Because these strands reside in the same strand space, our analysis 
includes possible attacks resulting from interactions between the UAF registration and authentication protocols.
We limit ourselves to three entities that directly create and process UAF protocol messages: client, server, and authenticator.

Instead of modeling authenticator strands, we embed the authenticator with the client, as is the case for a bound first-factor authenticator in UAF.
In our adversarial model, we do not consider an adversary that manipulates traffic between the client and the authenticator---rather, we analyze traffic that transmits over, and potentially binds to, a TLS channel between the client and the server (see Section~\ref{sec:adversarial_model}). 

Because we do not permit our adversary to clone legitimate authenticators, we omit the counters \hbox{$reg$-$cntr$} and $cntr$, which UAF specifies to mitigate this type of attack.
We assume that a legitimate server will never generate the same $challenge$ twice, thereby enabling it to detect any replayed attestation from the authenticator.
We also assume that the adversary cannot compromise a legitimate authenticator's attestation or authentication keys. For this reason, we do not model the attestation certificate during registration.
In light of these assumptions, omitting the counters 
does not permit any additional attacks within our scope.
Also, the specification provides insufficient details about many difficult issues, including synchronization, that can arise with counters.

{\mysaveB}
\subsection{CPSA Baseline Model} 
{\mysaveC}
\label{ssecbaseline-model}

For comparison, we begin with a \textit{baseline (no-UAF) model} in which a \textit{client} authenticates with a \textit{server} over a TLS channel in the absence of UAF.
This authentication takes place using a traditional username and password.
We model $client$ and $server$ as a pair of complementary strands and assume that they attempt to establish a TLS channel prior to the client's authentication messages.
We define this model's terms and strands:
\begin{itemize}
    \item Let $username$ be the user's identifier.
    \item Let $pw$ be the user's password.
    \item Let $swk$ be the server's TLS write key.
    \item Let $cwk$ be the client's TLS write key.
\end{itemize}

We now specify the messages for the \textit{Baseline-NoUAF} model.
The notation $\{term\}_{key}$ represents a term encrypted under a key.

\begin{enumerate}
    \item $client \longrightarrow server$: $\{username, password\}_{cwk}$
    \item $client \longleftarrow server$: $\{``auth OK''\}_{swk}$ -.
\end{enumerate}

Because the version of TLS does not affect the properties of this protocol, 
we present and analyze a baseline model only using TLS 1.2.

{\mysaveB}
\subsection{CPSA Models of Registration} 
{\mysaveC}
\label{ssecreg-model}

To model general registration, we specify a pair of complementary strands \textit{client-reg} and \textit{server-reg} that exchange terms reflecting those in UAF registration.
We assume that these strands first exchange terms modeling a TLS 1.2 or TLS 1.3 handshake, depending on the model's needs.
\begin{itemize}
    \item Let $username$, $appid$, and $aaid$ be names: $username$ specifies the user's account name, $appid$ specifies the server's name, and $aaid$ specifies the authenticator's name.
    \item Let $K_{pub}$ be an asymmetric authentication key, and $K_{priv}$ be the inverse of $K_{pub}$.
    \item Let $h = \hbox{hash}(K_{priv}, username)$.
    \item Let $fc = \hbox{hash}(appid, challenge, tlsData)$, where $tlsData$ is channel binding data.
    \item Let $s = \{aaid, fc, K_{pub}\}_{A}$, where $A$ is the authenticator's attestation signing key.
    \item Let $swk$ and $cwk$ be TLS server and client write keys.
\end{itemize}
We specify the registration protocol:
\begin{enumerate}
    \item $server\text-reg \longrightarrow client\text-reg$: $\{username, appid, challenge\}_{swk}$
    \item $server\text-reg \longleftarrow client\text-reg$: $\{h, aaid, K_{pub}, fc, s\}_{cwk}$.
\end{enumerate}

To model specific versions of UAF registration, we vary the construction of $tlsData$ to reflect a model's channel binding, and derive $swk$ and $cwk$ in a manner consistent with the model's TLS version.

{\mysaveB}
\subsection{CPSA Models of Authentication} 
{\mysaveC}
\label{ssecauth-model}

We model authentication in a manner similar to how we modeled registration, 
altering several terms and defining two new complementary strands: $client\text-auth$ and $server\text-auth$.
As in registration, we assume these strands,
prior to engaging in UAF authentication,
to transmit and receive terms that reflect a TLS handshake appropriate for the specific model.
In contrast with our registration model, we redefine $s = \{fc, n\}_{K_{priv}}$, where $n$ is a random nonce for the authentication protocol session.
We now specify the authentication protocol:
\begin{enumerate}
    \item $server\text-auth \longrightarrow client\text-auth$: $\{appid, challenge\}_{swk}$
    \item $server\text-auth \longleftarrow client\text-auth$: $\{fc, n, s\}_{cwk}$.
\end{enumerate}

As before, we vary the construction of $tlsData$ and match $swk$ and $cwk$ to reflect different choices for channel bindings and TLS versions.


\clearpage
\section{Security Goals}
{\mysaveAA}
\label{sec:goals}

In the strand space model, we formalize several UAF security goals as an \textit{injective agreement}~\cite{lowe1997hierarchy} on a session context between two complementary strands in a UAF strand space.
First, we introduce UAF security goals relevant to our study.
Second, we define \textit{session context agreement} as an injective agreement on the values and relationships of a protocol's session context variables.
Injective agreement means, in part, that the two roles
agree on the values of the specified variables. In addition, informally, it means
that this agreement was achieved through a unique proper execution of the
protocol, and not through an adverse interaction with a DY adversary.
Third, we define session contexts and cryptographic assumptions for our UAF models.

{\mysaveB}
\subsection{UAF Security Goals}
{\mysaveC}

FIDO states several security goals for which TLS channel binding serves as an explicit security measure~\cite{baghdasaryan2022fido}: 
SG-1 (user authentication), SG-11 (forgery resistance), SG-12 (parallel session resistance), and SG-13 (forwarding resistance).
Table~\ref{tbl:uafgoals} lists these goals.
FIDO describes these security goals informally, without
explicitly referencing authenticity goals for specific session variables.
As such, FIDO's statement of these goals is insufficient for any rigorous evaluation of UAF's authenticity properties.

In Section~\ref{sec:analysis}, we rigorously formalize UAF's SG-1, SG-11, SG-12, and SG-13 security goals collectively as
agreements between the server and the client on several session variables: 
username, server identifier, application ID, challenge, or authenticator public key.
Specifically, we collectively formalize these goals as an
injective agreement on the variables of a session context between a pair of complementary strands in a UAF strand space.
Because UAF protocols run over a TLS connection, there are additional implicit variables on which honest, complementary communicants must agree: server random nonce, client random nonce, premaster secret (RSA key exchange), or DH exponents.
As explained in Section~\ref{ssec:agreement}, this injective agreement on session context between complementary client and server strands implies that UAF achieves its informally stated security goals.


{\mysaveB}
\subsection{Session Context Agreement}
{\mysaveC}

First, we define (1)~session contexts, (2)~origination assumptions, and (3)~unique and successful completion.
Second, we combine these ideas to generalize a session context-agreement security goal.

\subsubsection{Session Context.}
\label{ssec:context}

Session contexts comprise variables that express a single, unique session of a protocol, such as the identifiers of the communicating parties, cryptographic nonces that provide the session's freshness, and other values, such as challenges, on which participants completing a protocol session must agree.
A well-constructed session context specifies the required relationships between two complementary protocol roles for a correct execution, and can be challenging to define when a protocol features many terms with potentially complex relationships.

For any arbitrary two-party protocol $\pi$ with a corresponding strand space $\Sigma_{\pi}$, a session context $Ctx[\pi]$ is a list of facts on which a pair of complementary strands $s\in\Sigma_{\pi}$ and $s'\in\Sigma_{\pi}$ must agree.
Facts reflect equivalences between values (e.g., $s$ and $s'$ must map a nonce $n$ to the same value) or relationships (e.g., $s$ generates a nonce $n$ for which $s'$ holds $\hbox{hash}(n)$).
If all facts in $Ctx[\pi]$ hold for a pair of strands $s$ and $s'$, then $s$ and $s'$ satisfy $Ctx[\pi]$.
Within a strand space, it is possible for $s$ and $s'$ to take on the value of any strands, including the case where $s = s'$, or the case where one of $s$ and $s'$ is an intruder strand.

\subsubsection{Origination Assumptions.}
\label{ssec:origination}

Origination assumptions are secrecy and freshness assumptions of a protocol’s variables, such as keys,
nonces, or passwords.  In the strand space model, these assumptions place constraints on the values  
penetrator strands may know or learn. 
\textit{Non-originating} variables (typically used for private keys)
cannot appear (i.e., be carried) as terms in any message, 
and are unavailable to the penetrator to know or learn, but are available to honest strands for
encrypting information. Several origination assumptions are available to constrain the origination of
sensitive values.  
\textit{Uniquely originating} variables (typically used for nonces)
generate on a single honest strand at the node where they
are first carried in a message. 
\textit{Uniquely generated} variables (typically used for DH exponents)
generate on a single honest strand at the node
where they first appear, but are not necessarily carried in a message. 
As such, they begin unknown to the adversary and the adversary cannot know
the values of uniquely originating or generated variables until after the node where they first generate. 
\textit{Penetrator non-originating} variables (typically used for passwords)
are unknown to the adversary, but unlike
non-originating terms, can be carried in a message and can be learned by the adversary.

Often, different protocol roles are unable to share the same 
assumptions---for example, a client may know that they generate a fresh nonce, but from a server's perspective, the client may reuse a previous nonce.
As a result, a protocol might achieve injective agreement on a term from one perspective, but not from another perspective.

\subsubsection{Successful and Unique Completion.}
 
We define two necessary predicates for injective agreement on the facts of a session context: \textit{successful completion} and \textit{unique completion}.
Together, these properties define an injective agreement: if both properties hold, two complementary strands complete a unique run of the protocol with each other and agree on the resulting session context.

Honest strands derive their terms from protocol \textit{roles}.
For an ideal execution, a strand that belongs to a role $\mathcal{R}_{a}\in\Sigma_{\pi}$ completes the protocol $\pi$ with a complementary strand that belongs to a role $\mathcal{R}_{b}\in\Sigma_{\pi}$.

For successful completion, two strands $s$ and $s'$ must execute the protocol with each other and agree on the resulting session context $Ctx[\pi]$ under a set of origination assumptions for the role of $s$, which we define as $Orig[\mathcal{R}_{a}]$, where $\mathcal{R}_{a}$ is the role of $s$.
Without an additional property specifying the protocol execution as unique, 
successful completion alone yields a non-injective agreement on the session context $Ctx[\pi]$.

In cases where $s$ transmits the final message, $s$ completes without assurance that $s'$ fully executes.
As a result, a bundle $\mathcal{B}$ may incorporate a strand $s'$ with the final node missing, which is a common occurrence when modeling protocols.
To describe this scenario, let $\mathcal{B}$-height be the node-height (number of terms transmitted and received) of a strand that appears in $\mathcal{B}$.
For any natural number $i$, a strand of integer height $i$ has sent or received at least $i$ of its terms.

\begin{definition}[Successful Completion]
\label{def:sc}
\hfill
\\
Let $\pi$ be any two-party protocol;
let $\mathcal{B}$ be any bundle in any strand space $\Sigma_{\pi}$; 
let $\mathcal{R}_{a}$ and $\mathcal{R}_{b}$ be roles in $\Sigma_{\pi}$;
let $s$ and $s^{\prime}$ be any pair of complementary strands in $\Sigma_{\pi}$;
let $Ctx[\pi]$ be any session context for $\pi$; and 
let $Orig[\mathcal{R}_{a}]$ be any set of origination assumptions for $s$.
Let $i$ and $j$ be any natural numbers.

$Succ(\mathcal{R}_{a}, \mathcal{R}_{b}, Ctx[\pi], Orig[\mathcal{R}_{a}], i ,j)$ \textit{if and only if (iff)} for all bundles $\mathcal{B}$ that yield from $\Sigma_{\pi}$ and strands $s \in \mathcal{R}_{a} \cap \mathcal{B}$, 
there exists a strand $s^{\prime} \in \mathcal{R}_{b} \cap \mathcal{B}$ such that under $Orig[\mathcal{R}_{a}]$, $s$ with $\mathcal{B}$-height $i$ and $s^{\prime}$ with $\mathcal{B}$-height $j$ satisfy $Ctx[\pi]$.
\end{definition}

Unique completion asserts that, if $\mathcal{B}$ contains executions of the strands $s$, $s'$, and $s''$ such that all strands agree on $Ctx[\pi]$, strands $s'$ and $s''$ must share role $\mathcal{R}_{b}$.
Because they share a role and agree on the same context, for the analysis we can assume $s' = s''$ and $s'$ is not an intruder strand.
Based on this observation, we can assert that only $s$ and $s'$ agree on $Ctx[\pi]$---in other words, 
$s$ and $s'$ are a unique completion of the protocol.

\begin{definition}[Unique Completion]
\label{def:uc}
\hfill
\\
Let $\mathcal{R}_{a}$ be any role in $\Sigma_{\pi}$;
let $Ctx[\pi]$ be any session context for $\pi$; and 
let $Orig[\mathcal{R}_{a}]$ be any set of origination assumptions for $s$.
Let $i$ and $j$ be any natural numbers.

$Uniq(\mathcal{R}_{a}, Ctx[\pi], Orig[\mathcal{R}_{a}], i ,j)$ iff for all roles $\mathcal{R}_{b}, \mathcal{R}_{c}$ in $\Sigma_{\pi}$, \\ $Succ(\mathcal{R}_{a}, \mathcal{R}_{b}, Ctx[\pi], Orig[\mathcal{R}_{a}], i ,j) \wedge Succ(\mathcal{R}_{a}, \mathcal{R}_{c}, Ctx[\pi], Orig[\mathcal{R}_{a}], i ,j) \implies$ $\mathcal{R}_{b} = \mathcal{R}_{c}$.
\end{definition}

\subsubsection{Session Context Agreement.}
\label{ssec:agreement}

We now define \textit{session context-agreement} (Goal~\ref{goal:context-agreement}) as a conjunction of unique and successful completion.
Using Goal~\ref{goal:context-agreement}, which we model explicitly in CPSA, we prove or disprove injective agreement of $Ctx[\pi]$ between the roles $\mathcal{R}_{a}$ and $\mathcal{R}_{b}$ under origination assumptions $Orig[\mathcal{R}_{a}]$ within a strand space $\Sigma_{\pi}$.
For any strand space $\Sigma_{\pi}$, if CPSA terminates and produces a counterexample to Goal~\ref{goal:context-agreement}, the counterexample describes an attack on the protocol.
In Section~\ref{sec:analysis}, we analyze our baseline and UAF models for session context agreement.

As a comprehensive authentication goal, Goal~\ref{goal:context-agreement} implies 
all of the weaker, informal UAF security goals (see Table~\ref{tbl:uafgoals}) in accordance with the Lowe authentication hierarchy~\cite{lowe1997hierarchy}.
Models that satisfy Goal~\ref{goal:context-agreement} express a one-to-one correspondence between the client and the server's UAF session context variables, implying SG-1.
Additionally, the SG-11, SG-12, and SG-13 UAF goals specify resistance against an adversary that impersonates a legitimate user via executing parallel sessions, forging messages, or forwarding existing messages.
Because Goal~\ref{goal:context-agreement} assumes an unbounded number of sessions, Goal~\ref{goal:context-agreement} implies SG-11.
Because the UAF strand space includes penetrator strands that forge and forward messages, 
Goal~\ref{goal:context-agreement} also implies SG-12 and SG-13.

\begin{goal}[Session Context Agreement]
\label{goal:context-agreement}
\hfill
\\
Let $\Sigma_{\pi}$ be any strand space;
let $\mathcal{R}_{a}$ and 
$\mathcal{R}_{b}$ be any complementary roles in $\Sigma_{\pi}$;
let $Orig[\mathcal{R}_{a}]$ be any set of origination assumptions for role $\mathcal{R}_{a}$;
and let $Ctx[\pi]$ be any session context for $\pi$.
For any role $\rho$, let $height(\rho)$ be the height of any fully executing strand in role $\rho$.
Let $i$ and $j$ be natural numbers such that $i = height(\mathcal{R}_{a})$ and $j \leq height(\mathcal{R}_{b})$.
\\\\
Session context agreement means $\exists j$ such that \\
$Succ(\mathcal{R}_{a}, \mathcal{R}_{b}, Ctx[\pi], Orig[\mathcal{R}_{a}], i, j) \wedge Uniq(\mathcal{R}_{a}, Ctx[\pi], Orig[\mathcal{R}_{a}], i, j)$.
\end{goal}

{\mysaveB}
\subsection{UAF Session Contexts}
{\mysaveC}

To define UAF session contexts, we identified key terms on which honest communicants must agree to achieve the UAF security goals, including terms crucial for a correct TLS session.
For comparison, we include a session context for the baseline model.
Because UAF session contexts depend on TLS session contexts, we begin by defining TLS session contexts.

\subsubsection{TLS Session Contexts.}

We first specify a session context for TLS-1.2-RSA, followed by a session context for TLS using ephemeral DH key exchange (TLS-1.2-DH, TLS-1.3).
The TLS contexts are interchangeable parameters for the baseline context and the UAF contexts that follow.
Both TLS contexts specify agreement for the client and server random nonces.
The TLS-1.2-RSA context specifies an additional condition on which communicants must agree involving the premaster secret.
To reflect a successful DH key exchange, the TLS-DH context specifies that communicants must agree on each communicant's private DH exponent.
It is possible for strands to agree on values that they do not know, such as the private exponent corresponding to a communicant's public DH value.

\begin{context}[TLS1.2-RSA]
\label{ctx:tls-1.2-rsa}
\hfill
\\
Let $s$ and $s'$ be complementary strands that carry out a TLS 1.2 handshake with a RSA key exchange.
For any pair $(s, s')$, the following facts hold:
strands $s$ and $s'$ agree on the identifiers for the session's server ($server$) and the certificate authority ($ca$) that signed the server's certificate.
Strands $s$ and $s'$ agree on the client's random nonce ($cr$), the server's random nonce ($sr$), and the client-generated premaster secret ($pms$).
\end{context}

\begin{context}[TLS-DH]
\label{ctx:tls-dh}
\hfill
\\
Let $s$ and $s'$ be complementary strands that carry out a TLS 1.2 or TLS 1.3 handshake with an ephemeral DH key exchange.
For any pair $(s, s')$, the following facts hold:
strands $s$ and $s'$ agree on the identifiers for the session's server ($server$) and the certificate authority ($ca$) that signed the server's certificate.
Strands $s$ and $s'$ agree on the client's random nonce ($cr$), the server's random nonce ($sr$), the client's DH exponent ($x$), and the server's DH exponent ($y$).
\end{context}

\subsubsection{Baseline Session Context.}

Incorporating our TLS contexts, we now specify the baseline session context (see Session Context~\ref{ctx:baseline}) for traditional, password-based authentication across a TLS channel.
In 2024, the baseline reflects a still prevalent manner by which clients authenticate to websites and other web-based services.
Successful authentication between a client and a server in this baseline implies that a client strand and a server strand agree on the user's identifier, the server's identifier (e.g., URL), the user's password, and the cryptographic details of the underlying TLS session.

\begin{context}[No-UAF]
\label{ctx:baseline}
\hfill
\\
Let $s$ and $s'$ be any strands in the \textit{Baseline-NoUAF} strand space.
For any pair $(s, s')$, the following facts hold:
strands $s$ and $s'$ agree on a TLS session context [TLS-1.2-RSA or TLS-DH].
$s$ and $s'$ agree on the terms $username$, $server$, and $pw$,
where $pw$ is the password.
\end{context}

\subsubsection{UAF Registration Session Context.}

We now define the UAF registration session context.
Similar to the baseline session context, this context require strands to agree on an appropriate TLS session context---the TLS context varies depending on which version of TLS our models assume.
Successful registration requires client and server strands participating in the same session to agree on the user's identifier, the server's cryptographic challenge, the server's identifier, and the $appId$.

\begin{context}[UAF Registration]
\label{ctx:registration}
\hfill
\\
Let $s$ and $s'$ be any strands in any UAF strand space.
For any pair $(s, s')$, the following facts hold:
strands $s$ and $s'$ agree on a TLS session context.
$s$ and $s'$ agree on the terms $challenge$, $username$, $server$, and $appid$.

\end{context}

\subsubsection{UAF Authentication Session Context.}

Because the authentication protocol is simpler, we define UAF authentication session context
as a subset of the registration context's facts.
Specifically, a client and a server strand for authentication must agree only on the server identifier, challenge, and $appId$.

\begin{context}[UAF Authentication]
\label{ctx:authentication}
\hfill
\\
Let $s$ and $s'$ be any strands in any UAF strand space.
For any pair $(s, s')$, the following facts hold:
$s$ and $s'$ agree on a TLS session context.
$s$ and $s'$ agree on the terms $challenge$, $server$, and $appid$.
\end{context}

{\mysaveB}
\subsection{UAF Origination Assumptions}
{\mysaveC}

We specify \textit{origination assumptions} (see Section~\ref{ssec:origination}) from the perspectives of the client and server roles for the baseline, UAF registration, and UAF authentication protocols.
In Section~\ref{sec:analysis}, we use these assumptions to analyze the baseline, UAF authentication, and UAF registration protocols.
Several assumptions apply only to a subset of the models we analyze, and carry labels to indicate such.

\begin{assumptions}[Client Assumptions]
\label{per:client}
\hfill
\\
For any initial client strand $\mathcal{C}$ in any bundle, we make the following assumptions:\\
(1)~Private keys of legitimate parties are unknown (non-originating) to the adversary and are never transmitted on the network.\\
(2)~$\mathcal{C}$ generates (uniquely originates) a fresh client-random nonce for TLS, and [for TLS-1.2-RSA] a fresh premaster secret or [for TLS-1.2-DH or TLS-1.3] a fresh DH value.\\
(3) [Baseline model only] The adversary does not know the client's password $pw$ (penetrator non-originating).\\
\end{assumptions}

\begin{assumptions}[Server Assumptions]
\label{per:server}
\hfill
\\
For any initial server strand $\mathcal{S}$ in any bundle, we make the following assumptions:\\
(1) Private keys of legitimate parties are unknown to the adversary and are never transmitted on the network.\\
(2) $\mathcal{S}$ generates a fresh server-random nonce for TLS, and [for TLS-1.2-DH or TLS-1.3] a fresh DH value.\\
(3) [Baseline model only] The adversary does not know the client's password $pw$ (penetrator non-originating).\\
(4) [UAF models only] The server uniquely originates the $challenge$ for each session.\\
(5) [UAF registration only] The server uniquely originates the client's username.\\
\end{assumptions}


\clearpage
\section{CPSA Analysis}
{\mysaveAA}
\label{sec:analysis}

We overview our analysis method and then state and prove session context-agreement theorems for our UAF registration and authentication models.

{\mysaveB}
\subsection{Method} 
{\mysaveC}
\label{ssec:method}

We analyze UAF by 
(1)~modeling UAF authentication and registration protocols in the strand space model, 
(2)~for each protocol, defining session contexts, 
(3)~specifying origination assumptions for each of the protocol roles, 
(4)~using CPSA, search the strand space models to produce ``shapes'' (see Section~\ref{ssecCPSA}), 
(5)~using the shapes, prove these goals true or false.
Our analysis builds from the models in Section~\ref{sec:cpsa-model}, 
session contexts from Section~\ref{ssec:context}, 
origination assumptions for the client and server role perspectives from Section~\ref{ssec:origination}, and successful session context-agreement (Goal~\ref{goal:context-agreement}).

In this section, we prove or disprove session context agreement by using CPSA 
to search exhaustively for unique shapes within a strand space.
Each search assumes the perspective of a role and its corresponding origination assumptions. Each search outputs shapes that support or serve as a counterexample to a session context agreement security goal.
When searching, CPSA includes intruder strands in strand spaces.

For comparison, our analysis includes baseline authentication without UAF (Section~\ref{sec:cpsa-model}).


{\mysaveB}
\subsection{Analysis of Registration} 
{\mysaveC}
\label{ssec:reg-analysis}

We specify and prove theorems for UAF registration.
Table~\ref{tbl:goals} summarizes our results. 

Theorems \ref{theorem:uaf-reg-server-false} and \ref{theorem:uaf-reg-server-true} address UAF registration from the perspective of a server.
Theorem~\ref{theorem:uaf-reg-server-false} asserts 
that the session context-agreement Goal~\ref{goal:context-agreement} is false.
Theorem~\ref{theorem:uaf-reg-server-true} applies to \textit{UAF-Endpoint-TLS1.2-DH}, \textit{UAF-ServerCert-TLS1.2-DH}, and \textit{UAF-Exporter-TLS1.3}, for which Goal~\ref{goal:context-agreement} holds under the server's perspective and origination assumptions.
We refer to multiple models that model the same channel binding (e.g., \textit{UAF-NoBiding-TLS1.2-RSA} and \textit{UAF-NoBinding-TLS1.3}) by substituting a wildcard (*) for the TLS version (e.g., \textit{UAF-NoBiding-*}).

\begin{theorem}
\label{theorem:uaf-reg-client}
For any initial client-reg role holding the client origination assumptions, and 
for any complementary server-reg role in the [\textit{UAF-NoBinding-*, UAF-TokenBinding-*, UAF-ChannelId-*, UAF-Endpoint-*, UAF-ServerCert-*, UAF-Exporter-TLS1.3}] strand spaces, Goal~\ref{goal:context-agreement} is true.
\end{theorem}

\begin{proof}[{\bf Proof} (By enumeration).]
For each UAF model, CPSA terminates and discovers a single shape that reflects an ideal execution between a client and a server, satisfying Goal~\ref{goal:context-agreement}.
Because CPSA terminated and discovered all essentially different shapes for these models, no counterexample exists.
\end{proof}

\begin{theorem}
\label{theorem:uaf-reg-server-false}
For any initial server-reg role holding the server origination assumptions, and 
for any complementary client-reg role in the [\textit{UAF-NoBinding-*, UAF-TokenBinding-*, UAF-ChannelId-*, UAF-Endpoint-TLS1.2-RSA, UAF-ServerCert-TLS1.2-RSA}] strand spaces, Goal~\ref{goal:context-agreement} is false.
\end{theorem}

\begin{proof}[{\bf Proof} (By counterexample).]
For each model, finds a shape that contradicts Goal~\ref{goal:context-agreement}.
Moreover, CPSA fails to find any shape that satisfies Goal~\ref{goal:context-agreement}.

Under the server assumptions, each UAF-NoBinding-* model produces a shape that illustrates a similar issue: because the server fails to bind its challenge to a session, the adversary may reissue the server's challenge to a client responding in a different session.
The server binds the challenge to the relying party's \textit{appId}, which is not a unique value. As we discuss in Section~\ref{sec:adversarial_model}, is possible for an adversary to spoof the \textit{appId} to a client.
An additional shape when using TLS with RSA indicates that the server additionally has no assurance that the client generates a unique premaster secret, enabling the adversary to compromise the channel.

Similarly, the UAF-TokenBinding-* and UAF-ChannelId-* models produce shapes that fail to satisfy the goal.
Because the server does not assume that the client's binding key is unique to the TLS session with the server, the server has no assurance that the adversary did not reissue the challenge in a separate UAF session.
The server cannot know if a client reused a binding key or if the adversary compromised the key.

The \textit{UAF-Endpoint-TLS1.2-RSA} and \textit{UAF-ServerCert-TLS1.2-RSA} models, in which the client binds the challenge parameters to the server's certificate, present several shapes that illustrate a more subtle issue: the binding does not bind to a specific session of UAF registration.
While there is some guarantee that the client intends to communicate with the server, it is possible for the client to confuse the session for which the server issues a challenge.
As in the other models, this issue arises from the server failing to bind the challenge to the underlying TLS session.
\end{proof}

\begin{theorem}
\label{theorem:uaf-reg-server-true}
For any initial server-reg role holding the server assumptions, and 
for any complementary client-reg role in the [\textit{UAF-Endpoint-TLS1.2-DH, UAF-ServerCert-TLS1.2-DH, UAF-Exporter-TLS1.3}] strand spaces, Goal~\ref{goal:context-agreement} is true.
\end{theorem}

\begin{proof}[{\bf Proof} (By enumeration).]
Each model produces a single shape that satisfies Goal~\ref{goal:context-agreement}.
\end{proof}

Each of the he \textit{UAF-Endpoint-TLS1.2-DH}, \textit{UAF-ServerCert-TLS1.2-DH}, and \textit{UAF-Exporter-TLS1.3} models 
overcomes the limitations of the other models.
Because these models use TLS with DH key exchange (TLS1.2-DH and TLS1.3), they ensure that both the client and the server contribute confidential freshness to the TLS channel.
The endpoint, server certificate, and exporter bindings ensure that the challenge binds to an authenticated entity---either to the server's identity, or the cryptographic context of the TLS channel itself.

{\mysaveB}
\subsection{Analysis of Authentication} 
{\mysaveC}
\label{ssec:auth-analysis}

Our analysis of UAF authentication is similar to 
that of registration because in both protocols, the server issues a challenge and receives a similar response.
Because the server issues the challenge in the same manner, it shares the main issue with the registration protocol: the challenge does not bind to a unique session in a manner that the client can verify.
Table~\ref{tbl:goals} includes columns summarizing our security goal theorems for our authentication roles.

\begin{theorem}
\label{theorem:uaf-auth-client}
For any initial client role holding the client assumptions,
and for any complementary server role in the \textit{[UAF-NoBinding-*, UAF-TokenBinding-*, UAF-ChannelId-*, UAF-Endpoint-*, UAF-ServerCert-*, UAF-Exporter-TLS1.3]} strand spaces, Goal~\ref{goal:context-agreement} is true.
\end{theorem}

\begin{proof}[{\bf Proof} (By enumeration).]
As for the client perspective in UAF registration, CPSA terminates and discovers only a single shape for each model, which shape illustrates correct execution, satisfying Goal~\ref{goal:context-agreement}.
\end{proof}

\begin{theorem}
\label{theorem:uaf-auth-server-false}
For any initial server-auth role holding the server assumptions, and 
for any complementary client-auth role in the [\textit{UAF-NoBinding-*, UAF-TokenBinding-*, UAF-ChannelId-*, UAF-Endpoint-TLS1.2-RSA, UAF-ServerCert-TLS1.2-RSA}] strand spaces, Goal~\ref{goal:context-agreement} is false.
\end{theorem}

\begin{proof}[{\bf Proof} (By counterexample).]
For each model, CPSA finds at least one counterexample.
As in the analysis of UAF registration, in the counterexamples, the server has no guarantee that its challenge remains within the server's session with the client.
In the unbound models, this issue results from a lack of binding between the challenge and the session outside of the \textit{appId}, which is an inadequate binding.
For the client-oriented channel bindings (token binding, channel ID), the server has no assurance that the client's key material for the binding is unique to the session.
For the server-oriented channel bindings (endpoint, server certificate), the server cannot determine to which TLS session between the server and the client the challenge binds.
In RSA-based models, the server has no assurance that the client contributes confidential freshness to the session.
\end{proof}

\begin{theorem}
\label{theorem:uaf-auth-server-true}
For any initial server-auth role holding the server assumptions, and 
for any complementary client-auth role in the [\textit{UAF-Endpoint-TLS1.2-DH, UAF-ServerCert-TLS1.2-DH, UAF-Exporter-TLS1.3}]strand spaces, Goal~\ref{goal:context-agreement} is true.
\end{theorem}

\begin{proof}[{\bf Proof} (By enumeration).]
As in Theorem~\ref{theorem:uaf-reg-server-true}, implementing exporter channel binding or binding to the server's authenticated identity when using TLS1.2-DH satisfies Goal~\ref{goal:context-agreement} for the UAF authentication server role.
\end{proof}
\clearpage
\section{Discussion}
{\mysaveAA}
\label{sec:discussion}

We now discuss five issues that arise in our work: 
(1)~the necessity of binding cryptographic values at their origin,
(2)~implications of making channel binding optional in UAF, 
(3)~dealing with proxies,
(4)~the value of formal-methods analyses, and 
(5)~open problems.

{\mysaveB}
\subsection{Binding Values at Their Origin} 
{\mysaveC}
\label{sec:origin-binding}

In a prudent protocol design, every protocol role must bind cryptographic values to a session's context.
Binding session data as early as possible---ideally at the origin of the data---reduces an adversary's opportunities to produce protocol interactions.
To build a session context, we suggest incorporating an underlying cryptographic session, incorporating additional attributes, such as the protocol's name, message sequence numbers, unique session identifiers, communicant identifiers, and any other values unique to the session.
Many protocol interactions arise from an inability to authenticate sensitive cryptographic values, resulting from a lack of binding. 
Often, the vulnerability can be mitigated by binding to a context appropriately.

In previous versions of our models, we included a ``dual binding'' model in which the server, in addition to the client, binds the challenge to the TLS channel using a TLS1.3 exporter binding.
From our analysis, we found this model is equivalent to our existing \textit{UAF-Exporter-TLS1.3} model, which satisfies each of our security goals.
Although we omit this model for redundancy, 
binding the challenge at the server would constitute best practice,
because doing so would make explicit 
for which UAF session the server generates the challenge.

{\mysaveB}
\subsection{UAF's Optional Channel Binding} 
{\mysaveC}
\label{sec:optional_binding}

In 2018, the FIDO Alliance intentionally specified these bindings as an optional feature of UAF, citing inadequate support for channel bindings and challenges facing perimeter proxies~\cite{technote2018}:
``$\dots$ the addition of channel-binding information in FIDO assertions is optional: not all client platforms support Channel ID or Token Binding, and even if a client has all the necessary support for channel-binding, it might make sense not to enforce channel-binding.''
To promote adoption of the UAF standard, FIDO specified channel binding as an optional feature.

From a cryptographer's perspective, without channel binding,
the UAF authentication protocol is vulnerable to a DY adversary: 
the adversary can exploit the lack of channel binding to launch a MitM attack, masquerading as an honest client to a server, without access to the client's authenticator.
To mitigate this attack, the protocol must require channel binding to an appropriate TLS session (e.g., TLS1.3).

From a policy perspective, requiring channel binding in the UAF standard presents practical problems: entities unable to support the channel binding standard may favor a different system over UAF, and some of the TLS channel bindings are still in draft form.
Optional channel binding is a conscious tradeoff between security and ease of adoption, which supports the goal of reducing and eliminating traditional, password-based authentication.
Policy makers, adopters, and users, however, should be aware that optional channel binding enables an adversary to create potentially serious protocol interactions in FIDO UAF authentication.

Additionally, the specification permits a server to disregard incorrect bindings to facilitate clients that bind to other channels (e.g., when communicating through a perimeter proxy, as explained in Section~\ref{ssec:binding}).
From a cryptographer's perspective, a server that fails to verify a binding enables harmful protocol interactions.
From a policy perspective, supporting clients that communicate through proxies facilitates the adoption of UAF but potentially enables hostile proxies (adversaries) to attack the protocol.
Developing a mechanism for the server to verify proxy bindings, such as cryptographically binding the server's session with
the client's channel with the proxy,
may mitigate the weakness enabled by this policy.

{\mysaveB}
\subsection{Dealing with Proxies} 
{\mysaveC}
\label{ssec:dis-proxy}

A significant obstacle to UAF's channel binding are clients that communicate through TLS proxies, which clients bind to channels between themselves and the proxies rather than to channels between themselves and UAF servers.
Such proxies might be border proxies that inspect incoming and outgoing traffic, or load-balancing proxies that distribute traffic to different endpoints.
This scenario affects channel bindings that bind to exported key material 
or to the server's identity, because the client will bind to key material,
or to a proxy's certificate, that does not correspond to the server.

To mitigate this issue, UAF grants servers discretion to accept incorrect bindings under the assumption that the server determines the binding to be reasonable, such as a binding to a well-known proxy.
This discretion presents an opening for an adversary wishing to attack the protocol by reissuing UAF challenges (see Section~\ref{sec:attack}), for which the primary mitigation is channel binding.
Unable to determine a client's intent from an incorrect binding, a server risks accepting a fraudulent assertion from an adversary.

There exist attempts, such as token binding~\cite{ietf-tokbind-ttrp-08}, 
to augment existing channel bindings
with measures to enable a proxy to convey additional information 
about a client's channel binding to a server.
To our knowledge, no such measures currently exist for exporter bindings.
We believe that channel bindings that explicitly address communication over TLS proxies present the most attractive long-term solution.

{\mysaveB}
\subsection{The Value of Using Formal Methods} 
{\mysaveC}
\label{ssec:formal-methods}

For simple and complex protocols, formal-methods protocol analysis is a powerful tool that helps make explicit a protocol's security goals and assumptions, identify vulnerabilities resulting from design decisions, and verify corrections made to known issues.
Vulnerabilities resulting from an adversary manipulating unbounded runs of a protocol are challenging for unassisted human beings to identify, and can linger for many years prior to discovery.
Adversaries with capabilities of performing formal-methods analysis may discover and exploit such vulnerabilities, if no benevolent studies identify the issues first.

The benefits of formal methods are particularly valuable---and often neglected---during the design of new protocols, which often defer formal verification in favor of publishing imperfect standards that experts, researchers, and 
adversaries later analyze.
While there is tremendous value in analyzing existing protocols, the most powerful applications of formal methods are during a protocol's design process.
Particularly complex, large standards such as UAF may require multiple formal-methods studies of their components during the design process, and a single study may be insufficient to consider all important aspects of the protocol.

FIDO's failure to apply formal methods while designing UAF 
contributed to vulnerabilities in the adopted standard: 
previous formal methods studies of UAF uncovered numerous attacks, 
and we identify an attack against UAF's channel binding.
There is a strong need for increased partnership between protocol designers and 
labs capable of formal-methods verification of protocols throughout the protocol design process.
Efforts such as those by Cremers et al. \cite{cremers2017comprehensive} to analyze emerging standards in development (e.g., TLS1.3) illustrate significant value in analyzing protocols to identify and correct vulnerabilities early, rather than when they potentially impact millions of users.

{\mysaveB}
\subsection{Open Problems and Future Work} 
{\mysaveC}
\label{ssec:open}

Our work on FIDO UAF authentication raises two open problems: 
(1)~Develop channel bindings that bind to protocol endpoints and sessions, and encapsulate additional bindings to address cases where the client or server communicate with a proxy.
(2)~Explore channel binding in FIDO2~\cite{fido2}, a new version of FIDO building on UAF that expands the available types of authenticators and improves interoperability with existing standards.

As future work, we are developing automatic tools for cryptographically binding protocol messages to their context in a manner that provably satisfies session context agreement in the DY model.
\clearpage
\section{CPSA Model Code (Selected Examples)} 
{\mysaveAA}
\label{sec:code}
We provide key CPSA model code snippets from our FIDO CPSA models, organized as roles and skeletons.
For complete models, see Github~\cite{PALgithub}.
Figure~\ref{fig:cpsa_src_client} and Figure~\ref{fig:cpsa_src_server} show CPSA source code for specifying the \textit{client-reg} and \textit{server-auth} strands.
Figure~\ref{fig:cpsa_src_uaf_macros} provides examples of macros that we call in strands to model UAF.
Figure~\ref{fig:cpsa_src_bindings} lists macros that we use to model the various UAF binding methods.
Figure~\ref{fig:cpsa_src_client_skeleton} and Figure~\ref{fig:cpsa_src_server_skeleton} are examples of CPSA skeletons, which specify the cryptographic specific for querying an authenticity security goal.
As explained in Section~\ref{ssecCPSA}, our models
superficially resemble (but are not) executable source code.

\begin{figure}[h] 
\begin{Verbatim}
  (defrole client-reg
    (vars
      (username auth server ca appid name)
      (challenge text)
      (cr sr random32) ;; TLS nonces. cr: client random, sr: server random
      (pms random48) ;; TLS pre-master secret generated by the client.
      (authk akey) ;; Authenticator's asymmetric key for this registration.
    )
    (trace
      (TLS send recv pms cr sr server ca)
      (UAF_Reg_Unbound_TLS12 send recv username appid challenge auth authk 
        cr sr pms)
    )
  )
\end{Verbatim}
\caption{
CPSA specification of a \textit{client-reg} role in UAF registration without channel binding, using TLS 1.2.
A role consists of variable definitions and a trace expressing terms that a strand sends and receives from the network.
Variable definitions comprise s-expressions that list labels followed by a sort (type).
The macro \textit{TLS} generates client terms to establish a TLS session with the server.
The macro \textit{UAF\_Reg\_Unbound\_TLS12} generates client terms that carry out UAF registration using a specific binding method and TLS version.
To model alternative versions, we substitute the TLS macro to generate terms for TLS 1.2 or TLS 1.3, and substitute the UAF macro for the desired binding and TLS version.
To model the \textit{auth-reg} role, we reverse the sign of each term by altering the order of the \textit{send} and \textit{recv} inputs to the macros.
}
\label{fig:cpsa_src_client}
\end{figure}

\clearpage

\begin{figure} 
\vspace*{24pt}
\begin{Verbatim}
(defrole server-auth
(vars
  (auth server ca appid name)
  (challenge text)
  (cr sr random32) ;; TLS nonces. cr: client random, sr: server random
  (authk akey)
  (spk akey) ;; Server's public key.
  (x expt) ;; Client's DH secret for TLS 1.3.
  (y rndx) ;; Server's DH secret for TLS 1.3.
)
(trace
  (TLS recv send cr sr x y server spk (invk spk) ca 
    (HandshakeSecret (DHpublic x) y))
  (UAF_Auth_Exporter_TLS13 recv send appid challenge authk cr sr spk x y)
)
(uniq-gen y)
)
\end{Verbatim}
\caption{
CPSA specification of a \textit{server-auth} role in UAF authentication, using TLS 1.3 and TLS exporter channel binding.
Contrasting Figure~\ref{fig:cpsa_src_client}, the macro \textit{TLS} now generates terms for TLS 1.3, and the macro \textit{UAF\_Auth\_Exporter\_TLS13} generates the appropriate terms for this code snippet's scenario.
As before, swapping the \textit{send} and \textit{recv} arguments for the macros generates terms appropriate for a \textit{client-auth} strand.
To carry out TLS 1.3's DH key exchange, we specify DH secrets and, from the server's perspective, declare $y$ as uniquely generated for each protocol session. 
}
\label{fig:cpsa_src_server}
\end{figure}


\begin{figure} 
\vspace*{24pt}
\begin{Verbatim}
(defmacro (UAF_Registration d1 d2 username appid challenge auth authk 
    binding cwk swk)
  (^
    (d2 (UAF_RegChallenge username appid challenge swk))
    (d1 (UAF_RegResponse username appid challenge auth authk binding cwk))))

(defmacro (UAF_Authentication d1 d2 appid challenge authk binding cwk swk)
  (^
    (d2 (UAF_AuthChallenge appid challenge swk))
    (d1 (UAF_AuthResponse challenge appid authk binding cwk))))
)
\end{Verbatim}
\caption{
Core CPSA macros for generating UAF registration and authentication terms.
Each macro accepts as an input an event ordering ($d1$, $d2$) comprising an ordering of \textit{send} and \textit{recv} events.
By altering the ordering of \textit{send} and \textit{recv}, the macro generates terms for both server and client roles.
The symmetric keys $cwk$ and $swk$ input from a preceding TLS 1.2 or TLS 1.3 handshake.
The channel binding $binding$ derives from one of several binding macros (see Figure~\ref{fig:cpsa_src_bindings}) or the string literal ``unbound''.
Remaining parameters correspond to UAF variables (see Section~\ref{sec:cpsa-model}).
}
\label{fig:cpsa_src_uaf_macros}
\end{figure}


\begin{figure} 
\vspace*{24pt}
\begin{Verbatim}
(defmacro (UAF_ChannelID challenge client)
  (hash "uaf_channel_id" challenge (pubk client) (enc (pubk client) 
    (privk client))))

(defmacro (UAF_TokenBinding challenge client)
  (hash "uaf_token_binding" challenge (cat (pubk client) (enc (pubk client) 
    (privk client)))))

(defmacro (UAF_ServerEndPoint challenge server ca)
  (hash "uaf_server_endpoint" challenge (hash (Certificate server ca))))

(defmacro (UAF_ServerCert challenge server ca)
  (hash "uaf_server_cert" challenge (Certificate server ca)))
  
;;; TLS Exporter PRF. We approximate this with a hash function.
(defmacro (UAF_Exporter ms label cr sr ctx)
  (hash ms label cr sr ctx))
\end{Verbatim}
\caption{
CPSA macros for incorporating different channel bindings in \textit{tlsData}.
Each macro, which we name after the corresponding channel binding in the UAF standard, implements a cryptographic abstraction of that binding in CPSA. 
}
\label{fig:cpsa_src_bindings}
\end{figure}


\clearpage
\vspace*{-48pt}
\begin{figure}[!h] 
\begin{Verbatim}
;; Client-reg perspective skeleton.
(defskeleton uaf-unbound-tls12
  (vars
    (server ca name)
    (cr sr random32)
    (pms random48)
    (authk akey)
  )
  (defstrandmax client-reg
    (server server) (ca ca) (cr cr) (pms pms) (authk authk) (sr sr))
  (non-orig (privk ca) (privk server) (invk authk))
  (uniq-orig cr pms)
  (uniq-orig sr) ;; Prevents multiple instances of the same server 
                 ;; interacting with the client.
)

\end{Verbatim}
\vspace*{-12pt} 
\caption{
CPSA skeleton that instantiates a \textit{client-reg} strand and specifies the strand's cryptographic assumptions.
To prevent the intruder from learning the private keys of $ca$, $server$, we specify their private keys as non-originating.
Similarly, we specify the asymmetric inverse authentication key $authk$ as non-originating.
To guarantee freshness of TLS values, we specify the client nonce $cr$, the server nonce $sr$, and the premaster secret $pms$ as uniquely originating---the strand in this skeleton will generate fresh values for each run of the protocol.
}
\label{fig:cpsa_src_client_skeleton}
\end{figure}

\vfill

\begin{figure}[!h] 
\begin{Verbatim}
;; Server-reg perspective skeleton.
(defskeleton uaf-unbound-tls12
  (vars
    (server ca name)
    (sr random32)
    (challenge text)
    (authk akey)
  )
  (defstrandmax server-reg
    (server server) (ca ca) (sr sr) (challenge challenge) (authk authk))
  (non-orig (privk ca) (privk server) (invk authk))
  (uniq-orig challenge sr)
)
\end{Verbatim}
\vspace*{-8pt} 
\caption{
CPSA skeleton that instantiates a \textit{server-reg} strand and specifies the strand's cryptographic assumptions.
This server strand shares the non-origination assumptions with the client strand in Figure~\ref{fig:cpsa_src_client_skeleton}.
The server strand assumes freshness of $challenge$ and $sr$ because we specify these as uniquely originating.
}
\label{fig:cpsa_src_server_skeleton}
\end{figure}

\clearpage 
\setcounter{tocdepth}{3}
\tableofcontents 

\end{document}